\documentclass[a4paper,cleveref, autoref, thm-restate]{lipics-v2021}
\usepackage{amsmath}
\usepackage{amsfonts}
\usepackage{amsthm}
\usepackage{mathtools}
\usepackage{amssymb}
\usepackage{acronym}
\usepackage{algorithm}
\usepackage{algpseudocodex}
\usepackage{booktabs}
\usepackage{tabularx}
\usepackage{multicol}
\usepackage{siunitx}
\usepackage{color} 

\algrenewcommand\textproc{\texttt}

\nolinenumbers

\title{A Simple Representation of Tree Covering Utilizing Balanced Parentheses and Efficient Implementation of Average-Case Optimal RMQs}
\titlerunning{A Simple Tree Covering Utilizing BPs and Efficient Average-Case Optimal RMQs}
\date{}
\author{Kou Hamada}{The University of Tokyo, Japan}{zkouaaa@g.ecc.u-tokyo.ac.jp}{https://orcid.org/0009-0005-9046-9818}{}
\author{Sankardeep Chakraborty}{The University of Tokyo, Japan}{sankardeep.chakraborty@gmail.com}{}{}
\author{Seungbum Jo}{Chungnam National University, South Korea}{sbjo@cnu.ac.kr}{https://orcid.org/0000-0002-8644-3691}{}
\author{Takuto Koriyama}{The University of Tokyo, Japan}{}{}{}
\author{Kunihiko Sadakane}{The University of Tokyo, Japan}{sada@mist.i.u-tokyo.ac.jp}{https://orcid.org/0000-0002-8212-3682}{}
\author{Srinivasa Rao Satti}{Norwegian University of Science and Technology, Norway}{srinivasa.r.satti@ntnu.no}{https://orcid.org/0000-0003-0636-9880}{}
\authorrunning{K. Hamada, S. Chakraborty, S. Jo, T. Koriyama, K. Sadakane, and S.\,R. Satti}
\Copyright{Kou Hamada, Sankardeep Chakraborty, Seungbum Jo, Takuto Koriyama, Kunihiko Sadakane, and Srinivasa R. Satti}
\supplement{}
\supplementdetails[subcategory={Source Code}, cite={}, swhid={}]{Software}{https://github.com/zkou-ut/hyperrmq}

\EventEditors{Timothy Chan, Johannes Fischer, John Iacono, and Grzegorz Herman}
\EventNoEds{4}
\EventLongTitle{32nd Annual European Symposium on Algorithms (ESA 2024)}
\EventShortTitle{ESA 2024}
\EventAcronym{ESA}
\EventYear{2024}
\EventDate{September 2--4, 2024}
\EventLocation{Royal Holloway, London, United Kingdom}
\EventLogo{}
\SeriesVolume{308}
\ArticleNo{51}

\mathtoolsset{showonlyrefs=true}

\renewcommand{\epsilon}{\varepsilon}

\newcommand{\Order}{\mathrm{O}}
\newcommand{\order}{\mathrm{o}}

\DeclareMathOperator*{\argmin}{argmin}

\newcommand{\BPb}{\mathrm{BP}_\mathrm{b}}
\newcommand{\BPo}{\mathrm{BP}_\mathrm{o}}

\newcommand{\rank}{\texttt{rank}}
\newcommand{\ranko}{\rank_{\texttt{(}}}
\newcommand{\rankc}{\rank_{\texttt{)}}}
\newcommand{\select}{\texttt{select}}
\newcommand{\selecto}{\select_{\texttt{(}}}
\newcommand{\selectc}{\select_{\texttt{)}}}
\newcommand{\access}{\texttt{access}}
\newcommand{\open}{\texttt{open}}
\newcommand{\close}{\texttt{close}}
\newcommand{\enclose}{\texttt{enclose}_{\texttt{)}}}

\newacro{BP}[BP]{Balanced-Parenthesis}
\newacro{RMQ}[RMQ]{Range Minimum Query}
\newacro{LCA}[LCA]{Lowest Common Ancestor}
\newacro{DFS}[DFS]{Depth-First Search}
\newacro{BFS}[BFS]{Breadth-First Search}
\newacro{LCP}[LCP]{Longest Common Prefix}

\begin{document}
\maketitle

\begin{abstract}
Tree covering is a technique for decomposing a tree into smaller-sized trees with desirable properties, and has been employed in various succinct data structures. However, significant hurdles stand in the way of a practical implementation of tree covering: a lot of pointers are used to maintain the tree-covering hierarchy and many indices for tree navigational queries consume theoretically negligible yet practically vast space. To tackle these problems, we propose a simple representation of tree covering using a balanced-parenthesis representation. The key to the proposal is the observation that every micro tree splits into at most two intervals on the BP representation. Utilizing the representation, we propose several data structures that represent a tree and its tree cover, which consequently allow micro tree compression with arbitrary coding and efficient tree navigational queries. We also applied our data structure to average-case optimal RMQ by Munro et al.~[ESA 2021] and implemented the RMQ data structure. Our RMQ data structures spend less than $2n$ bits and process queries in a practical time on several settings of the performance evaluation, reducing the gap between theoretical space complexity and actual space consumption. For example, our implementation consumes $1.822n$ bits and processes queries in \SI{5}{\micro\second} on average for random queries and in \SI{13}{\micro\second} on average for the worst query widths. We also implement tree navigational operations while using the same amount of space as the RMQ data structures. We believe the representation can be widely utilized for designing practically memory-efficient data structures based on tree covering. 
\end{abstract}

\newpage

\section{Introduction}

With the explosive growth of data volumes, data structures that process data quickly and memory-efficiently are of paramount interest. Data structures that asymptotically achieve the information-theoretic lower bound and support operations efficiently are called succinct data structures. Starting with the work by Jacobson~\cite{Jacobson1989Space}, many succinct solutions have been proposed for various settings.

Trees have been thoroughly studied in the context of succinct representations due to their fundamental nature and wide applicability. One major subclass of trees is ordinal trees, which are rooted trees where the children of each node are ordered. The information-theoretic lower bound of representing ordinal trees with $n$ nodes is $2n - \Theta(\log n)$\footnote{In this paper, $\log$ and $\lg$ denote the logarithm of base $e$ and $2$, respectively.} bits~\cite{Munro2001Succinct}. Researchers have devised many succinct data structures for ordinal trees, such as the level-order unary degree sequence (LOUDS)~\cite{Jacobson1989Space}, the balanced parentheses (BP)~\cite{Munro2001Succinct}, and the depth-first unary degree sequence (DFUDS)~\cite{Benoit2005DFUDS}, to name a few. 
All these representations are efficient in practice~\cite{Arroyuelo2010SuccinctTrees,Navarro2014FullyFunctional}. Another important subclass of trees is cardinal trees. Cardinal trees are rooted trees where every node has a fixed number of labeled slots and every slot has a child node or is empty. In this paper, we consider the cardinal trees whose nodes have two slots, which are called binary trees. Munro and Raman~\cite{Munro2001Succinct} employed a bijection between binary trees with $n$ nodes and ordinal trees with $n + 1$ nodes to extend their BP representation for ordinal trees to binary trees, obtaining a succinct representation of binary trees using $2n + \order(n)$ bits.

All the representations of trees described above convert trees into sequences and translate tree navigational queries into operations on the sequences. Another promising approach is tree covering~\cite{Farzan2014UniformParadigm}, which decomposes trees into parts called micro trees.
Since the decomposition is suited for look-up tables, tree covering can support a variety of queries in constant time~\cite{Davoodi2014EncodingRangeMinima,Farzan2014UniformParadigm}. Also, its hierarchical structure is beneficial in designing succinct data structures involving trees~\cite{ChakrabortyJSS21,Chakraborty2021CliqueWidth,Davoodi2014EncodingRangeMinima,Farzan2014UniformParadigm,He2014Framework,Munro2021Hypersuccinct,Tsur2018representation}. For example, hypersuccinct trees~\cite{Munro2021Hypersuccinct}, which we discuss in Sec.\,\ref{subsec:hypersuccinct_trees}, encode micro trees with a Huffman code and achieve optimal compression for various tree sources.

Although tree covering powerfully facilitates designing succinct data structures for trees, the data structures based on tree covering tend to require numerous $\order(n)$-bit indexes, such as pointers and look-up tables. While the data structure supports a wide range of tree navigational queries, the number of indexes needed to support the queries increases as well. Thus, straightforward implementation of data structures based on tree covering is unlikely to be efficient in practice. To the best of our knowledge, there is no practical implementation of succinct data structures based on tree covering. Thus, a practical design of tree covering suitable for implementation is extremely desirable.

We illustrate the need for practical tree-covering data structures by discussing their application to the \ac{RMQ} problem. Given a static array $A$ of length $n$ consisting of totally ordered objects, an RMQ data structure supports the following queries efficiently: given two indices $i, j$ with $1 \le i \le j \le n$, return $\argmin_{i \le k \le j} A[k]$, i.e., the index of the minimum in the subarray of $A$ from the $i$-th element to the $j$-th element. The problem appears as a subroutine in many real-world applications, such as auto-completion~\cite{Hsu2013Completion}, data compression~\cite{Chen2008LZ}, and document retrieval~\cite{Sadakane2007TextRetrieval}. 

Fischer and Heun~\cite{Fischer2011Space} first designed a succinct \ac{RMQ} data structure using $2n + \order(n)$ bits, achieving the worst-case optimal space complexity.
As for implementation, Ferrada and Navarro~\cite{Ferrada2017ImprovedRMQ} provided a practical implementation of a succinct \ac{RMQ} data structure by utilizing the range min-max tree~\cite{Navarro2014FullyFunctional}. Their implementation consumes $2.1 n$ bits and takes 1--\SI{3}{\micro\second} per query. Baumstark et al.~\cite{Baumstark2017Practical} also gave another competitive implementation with a faster query time of about \SI{1}{\micro\second}.

On the other hand, if we assume that the input array for \ac{RMQ} is a random permutation, the lower bound for the expected space consumption drops below $2n$ bits: this lower bound, also known as the expected effective entropy, is $1.736 n + \order(n)$ bits~\cite{Golin2016Encoding, Kieffer2009Structural}. Although the asymptotic lower bound for the worst-case space complexity remains $2n$ bits, it opens up the possibility of designing \ac{RMQ} data structures that consume less than $2n$ bits on average.

Davoodi et al.~\cite{Davoodi2014EncodingRangeMinima} proposed an \ac{RMQ} data structure that uses $1.919 n + \order(n)$ bits on average and supports constant-time queries. Munro et al.~\cite{Munro2021Hypersuccinct} designed hypersuccinct trees and applied them to Cartesian trees, obtaining an \ac{RMQ} data structure that takes $1.736 n + \order(n)$ bits on average and $2n + \order(n)$ bits in the worst case while supporting constant-time queries. They also found that their \ac{RMQ} data structure uses $2 \lg \binom{n}{r} + \order(n)$ bits when applied to an array of length $n$ with $r$ increasing runs. The average space consumption of their \ac{RMQ} data structure is asymptotically optimal in both cases.

While both RMQ data structures theoretically improve the average space consumption from the succinct solution by Fischer and Heun~\cite{Fischer2011Space}, they employ the tree-covering technique, making straightforward implementation inefficient in practice. To the best of our knowledge, there is no implementation of such an RMQ data structure, nor is there any implementation that consumes less than $2n$ bits. Therefore, a practical representation of tree covering may lead to a space-efficient RMQ data structure that is unprecedented.

\subsection{Our contribution}

Our main contribution is the proposal of a simple representation of tree covering in the BP representation for both ordinal trees and binary trees. The representation is based on the observation that every micro tree splits into at most two intervals in the BP representation. Utilizing the representation, we propose several practical designs of succinct data structures for trees and their tree covers. Also, as an application, we present an optimized design of an average-case optimal \ac{RMQ} data structure based on hypersuccinct trees. 
We also implement \ac{RMQ} data structures using our tree-covering representation. In empirical evaluations, the implementations spend less than $2n$ bits and process queries in a practical time on several settings of the performance evaluation. Furthermore, we implement tree navigational operations while spending the same amount of space as the \ac{RMQ} data structures.

The remainder of this paper is structured as follows. Sec.\,\ref{sec:preliminaries} introduces the prerequisite knowledge. Sec.\,\ref{sec:bp_repr_of_tc} delivers a simple representation of tree-covering structure in the BP representation for ordinal trees and binary trees and offers memory-efficient tree-covering data structures. Sec.\,\ref{sec:optimization_RMQ} applies our practical data-structure design to RMQ data structures and discusses further optimization. Sec.\,\ref{sec:performance_evaluation} presents the performance evaluation. Sec.\,\ref{sec:conclusion} concludes the paper and provides future directions. 

\section{Preliminaries}
\label{sec:preliminaries}

\subsection{Balanced Parentheses}
\label{subsec:balanced_parentheses}

In this section, we introduce a balanced sequence of parentheses and some operations on it~\cite{Navarro2016Compact}, which will be useful in the \ac{BP} representation of trees. A sequence of parentheses is balanced if the following conditions are satisfied:
\begin{itemize}
    \item The sequence is of length $2n$, i.e., even, and consists of $n$ opening parentheses and $n$ closing parentheses.
    \item The sequence has $n$ pairs of matching parentheses: each pair has an opening parenthesis on the left and a closing parenthesis on the right, and when considering intervals whose endpoints are matching parentheses, any two of the intervals are disjoint or one contains the other.
\end{itemize}

We define some operations on a balanced sequence of parentheses. In the following three operations, each parenthesis is represented by its index in the sequence.
\begin{itemize}
    \item Given a closing parenthesis, the \open{} operation returns the opening parenthesis that matches the given closing parenthesis.
    \item Given an opening parenthesis, the \close{} operation returns the closing parenthesis that matches the given opening parenthesis.
    \item Given a closing parenthesis, the $\enclose{}$ operation returns the closing parenthesis such that the corresponding interval tightly encloses the interval of the given closing parenthesis.
\end{itemize}
Fig.\,\ref{fig:BP_operations} shows an example of the three operations above. Note that the $\enclose{}$ operation is different from the standard \textproc{enclose} operation~\cite{Navarro2016Compact}, which takes and returns opening parentheses.

\begin{figure}[hbtp]
    \centering
    \includegraphics[width=0.4\linewidth]{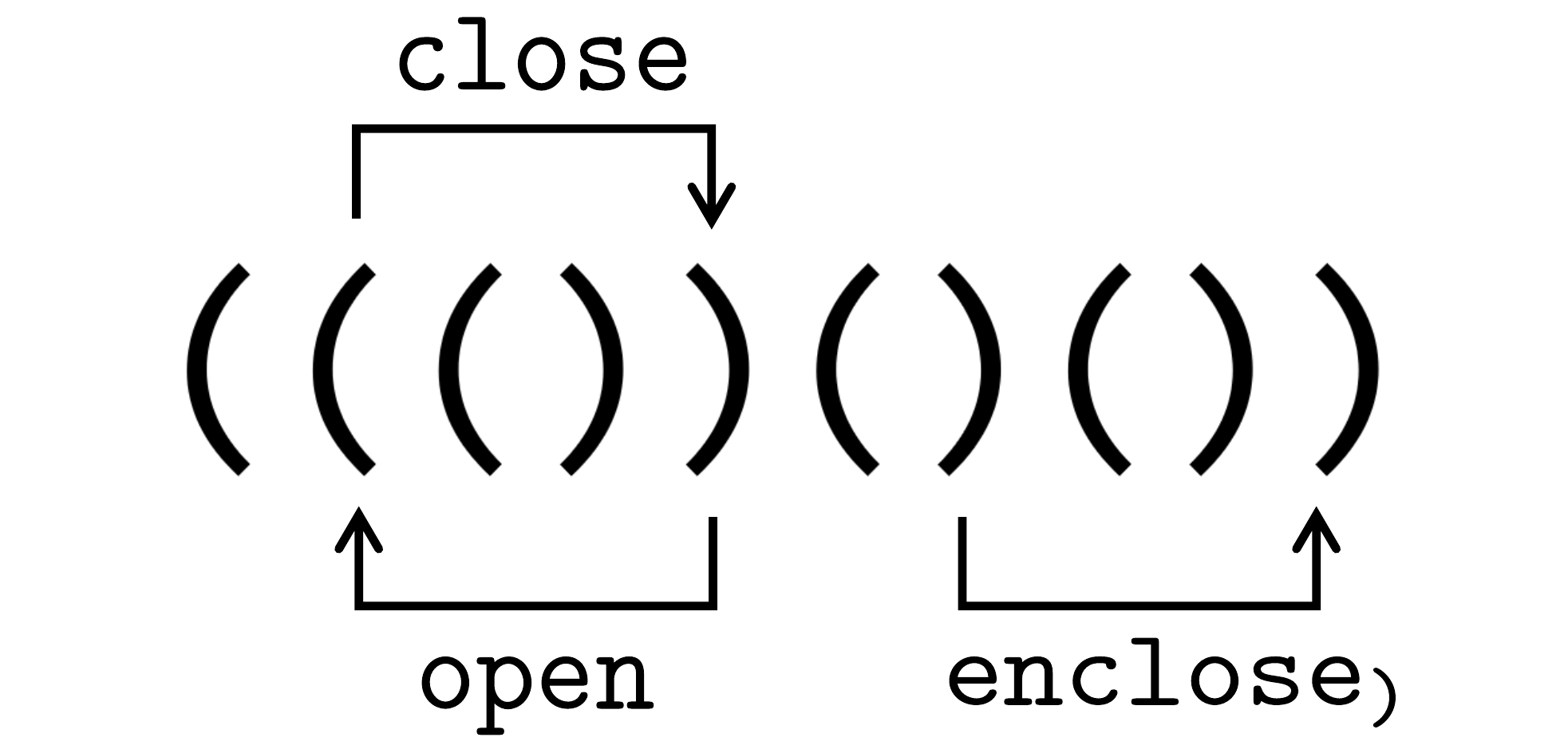}
    \caption{An example of the operations specific to a balanced sequence of parentheses.}
    \label{fig:BP_operations}
\end{figure}

We also define some useful operations that are not specific to a balanced sequence of parentheses.
\begin{itemize}
    \item Given an index of the sequence, the \access{} operation returns the parenthesis at the given index.
    \item Given an index $r$, the $\ranko$ and $\rankc$ operations respectively count the number of opening and closing parentheses among the first $r$ parentheses in the sequence.
    \item Given a number $j$, the $\selecto$ and $\selectc$ operations respectively find the $j$-th opening and closing parenthesis and return the index of the parenthesis in the sequence.
\end{itemize}

We define another operation called \textproc{rmq}, which is used to implement \textproc{lca} on the \ac{BP} representation of a binary tree. To explain \textproc{rmq}, we define $\textproc{excess}(k)$ to be $\ranko(k) - \rankc(k)$.
If the intervals formed by matching parentheses are half-open, i.e., if the intervals contain only their left endpoints and not their right endpoints, then the value $\textproc{excess}(k)$ corresponds to the number of intervals 
containing the $k$-th parenthesis. Using \textproc{excess}, we define \textproc{rmq} as follows:
\begin{itemize}
    \item Given indices $i$ and $j$, the \textproc{rmq} operation returns the index $k \in [i, j]$ that minimizes $\textproc{excess}(k)$; if there are multiple minima, it returns the index of the leftmost minimum.
\end{itemize}

We explain a succinct data structure for a balanced sequence of parentheses, called a range min-max tree~\cite{Navarro2014FullyFunctional}. Representing all the balanced sequences of $2n$ parentheses asymptotically requires $2n$ bits~\cite{Navarro2016Compact}. The range min-max tree spends $2n + \order(n)$ bits and supports a wide range of operations including the above operations. It is practically efficient in both time and space~\cite{Arroyuelo2010SuccinctTrees,Ferrada2017ImprovedRMQ}. The main idea of the range min-max tree is to split the BP into blocks and build a complete binary tree on the blocks. Each node of the complete binary tree corresponds to contiguous blocks and stores the values used to support queries. 

To support the queries discussed above, a simplified variant of the range min-max tree called the range min tree~\cite{10.1145/2656332} is sufficient. Each node stores two values: the change in and the minimum of excess.

\subsection{BP Representations of Trees}

\subsubsection{Definition and Properties}
\label{subsubsec:bp_def}

Here, we define the \ac{BP} representations of ordinal trees and binary trees~\cite{Davoodi2017SuccinctBinaryTree,Munro2021Hypersuccinct,Navarro2016Compact} and present some properties of the representations.

Before discussing the representations, we briefly describe the definition of ordinal trees and binary trees. An ordinal tree is a rooted ordered tree whose nodes have an ordered sequence of children, which can be possibly empty. A binary tree is a rooted ordered tree whose nodes have a left and right child, each of which can be empty. 
Note that, unlike ordinal trees, a node with a single left child and a node with a single right child are distinguished when regarded as binary trees. The BP representations described below reflect this difference between ordinal trees and binary trees.

\begin{definition}[\ac{BP} encoding of ordinal trees]
    \label{def:ordinal_tree_bp}
    The \ac{BP} encoding $\BPo(t)$ of an ordinal tree $t$ is defined recursively as follows: $\BPo(t) = \texttt{(} \cdot \BPo(t_1) \cdots \BPo(t_k) \cdot \texttt{)}$.
    Here, $k$ denotes the number of the children of the root, and $t_i$ denotes the subtree rooted at the $i$-th child of the root, respectively.
\end{definition}

\begin{definition}[\ac{BP} encoding of binary trees]
    \label{def:binary_tree_bp}
    The \ac{BP} encoding $\BPb(t)$ of a binary tree $t$ is defined recursively as follows:
    \begin{equation}
        \BPb(t) = 
        \begin{cases}
            \epsilon & \text{if $t$ is empty;} \\
            \texttt{(} \cdot \BPb(t_l) \cdot \texttt{)} \cdot \BPb(t_r) & \text{otherwise.}
        \end{cases}
    \end{equation}
    Here, $t_l$ and $t_r$ denote the subtrees whose roots are the left and right child of the root, respectively.
\end{definition}

\begin{remark}
    In both encodings, we correspond nodes to matching pairs of the BP sequence as follows: consider expanding the recursive definition of the BP sequence, and a node $v$ corresponds to the matching parentheses inserted when expanding the subtree rooted at $v$ during the recursion.
\end{remark}

Examples of the BP encodings are shown in Fig.\,\ref{fig:BP_example}. In both encodings, nodes correspond to matching pairs of parentheses, and subtrees correspond to balanced intervals.
Conversely, every parenthesis has a corresponding node. In particular, when discussing tree navigational queries on binary trees, we map nodes to the corresponding closing parentheses.

\begin{figure}[hbtp]
    \centering
    \begin{subfigure}{0.5\linewidth}
    \centering
    \includegraphics[width=0.95\linewidth]{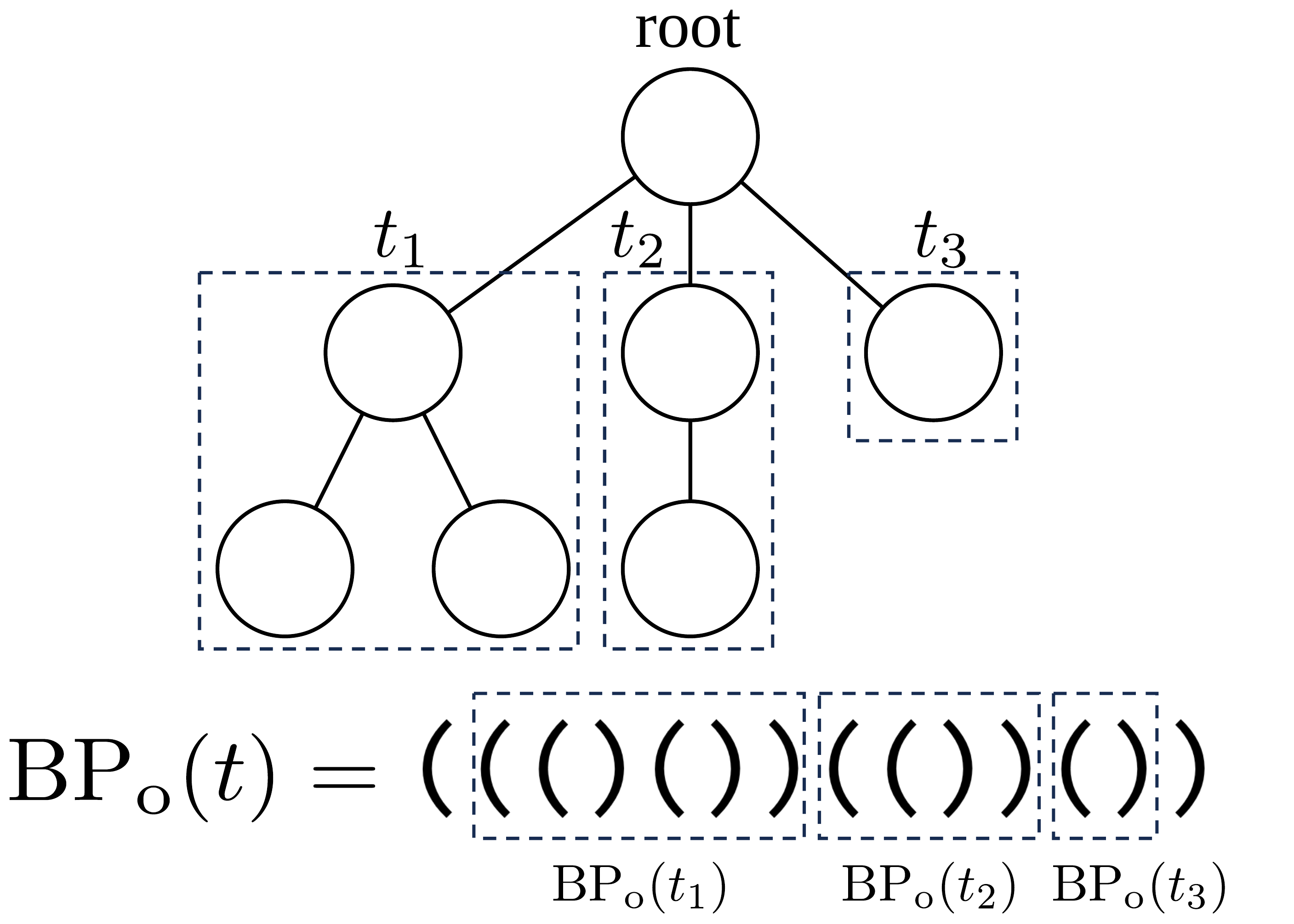}
    \caption{The BP encoding of an ordinal tree.}
    \end{subfigure}
    \hspace{0.03\linewidth}
    \begin{subfigure}{0.41\linewidth}
    \centering
    \includegraphics[width=0.95\linewidth]{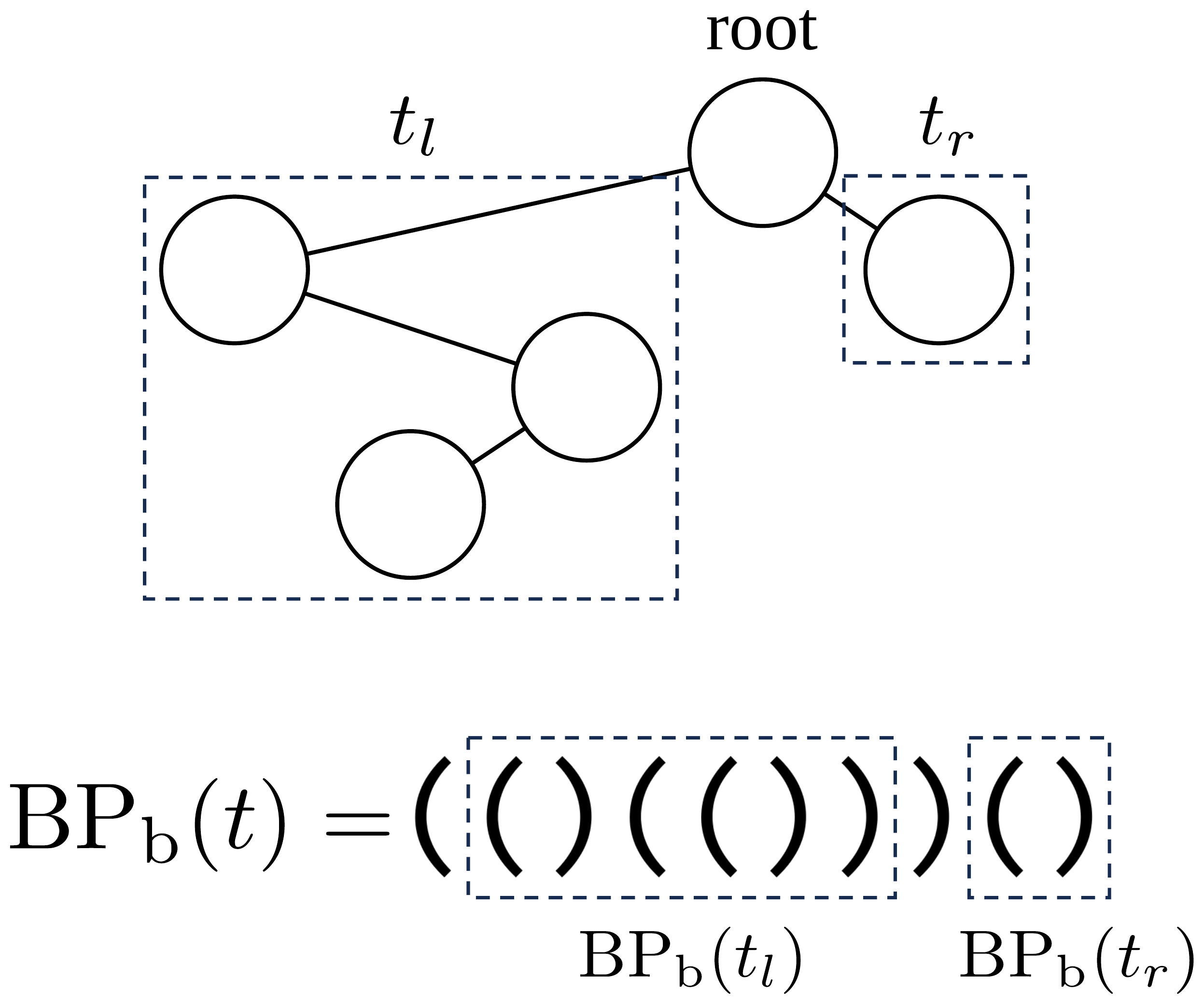}
    \caption{The BP encoding of a binary tree.}
    \end{subfigure}
    \caption{Examples of the BP encodings of ordinal trees and binary trees.}
    \label{fig:BP_example}
\end{figure}

As binary trees can be recovered from balanced sequences of parentheses, the function $\BPb(t)$ is a bijection from the set of all binary trees with $n$ nodes to the set of all the balanced sequences of $2n$ parentheses. Also, a balanced sequence of $2n$ parentheses can be mapped to an ordinal tree with $n + 1$ nodes by enclosing the sequence with a pair of parentheses and interpreting it as the $\BPo$ representation; this is a bijection from the set of all the balanced sequences of $2n$ parentheses to the set of all the ordinal trees with $n + 1$ nodes. It is worth noting that the composition of these bijections yields a known bijection between binary trees with $n$ nodes and ordinal trees with $n + 1$ nodes~\cite{Munro2021Hypersuccinct, Munro2001Succinct}.

The construction of the $\BPb$ sequence from a binary tree can be achieved by the following steps:
\begin{enumerate}
    \item Prepare an empty sequence.
    \item Visit the nodes in the \ac{DFS} order. If there are two children, visit the left child first. 
    \item For each node, append an opening parenthesis to the sequence when visiting the node for the first time; append a closing parenthesis immediately after visiting all nodes of the left subtree.
\end{enumerate}
The construction above implies the following proposition.

\begin{proposition}[Order of parentheses in BP of binary trees]
    \label{prop:order_of_paren}
    The $i$-th opening (resp. closing) parentheses in the BP sequence of a binary tree corresponds to the $i$-th node in the preorder (resp. inorder).
\end{proposition}
Prop.\,\ref{prop:order_of_paren} enables converting a node, its preorder, and its inorder to one another. The details are discussed in the next section.

\subsubsection{Tree Navigational Operations}
\label{subsec:bp_query}

In this section, we present how to support tree navigational operations on binary trees using the \ac{BP} operations described in Sec.\,\ref{subsec:balanced_parentheses}. Here, we only consider binary trees and the $\BPb$ representation.

We map nodes to the corresponding closing parentheses since it works well with inorder and \textproc{lca}. Also, we assume the BP is enclosed with a pair of parentheses as sentinels. These sentinels work in two cases: when \access{} is called with $0$ or $2n+1$, and when $\enclose{}$ is called on the parentheses such that the value of excess is zero. We assume the sentinels do not change the behavior of the other operations such as $\rank{}$ and $\select{}$.

\begin{algorithm}[hbtp]
\caption{Tree navigational operations that are possible on the $\BPb$ representation.}
\label{alg:navigation_on_bp}
\begin{multicols}{2}
\begin{algorithmic}[1]
\Function{preorder}{$v$}
\State \Return $\ranko(\Call{open}{v})$
\EndFunction

\Statex

\Function{preorderselect}{$i$}
\State \Return $\Call{close}{\selecto(i)}$
\EndFunction

\Statex

\Function{inorder}{$v$}
\State \Return $\rankc(v)$
\EndFunction

\Statex

\Function{inorderselect}{$i$}
\State \Return $\selectc(i)$
\EndFunction

\Statex

\Function{root}{\mbox{}}
\State \Return \close(1)
\EndFunction

\Statex

\Function{parent}{$v$}
\State $p \gets \Call{open}{v}$
\If{$p = 1$}
\State \Return $-1$
\EndIf
\If{$\Call{access}{p - 1} = \text{``\texttt{(}''}$}
\State \Return $\Call{close}{p - 1}$
\Else
\State \Return $p - 1$
\EndIf
\EndFunction

\Statex

\Function{leftchild}{$v$}
\State $p \gets \Call{open}{v}$
\If{$\Call{access}{p + 1} = \text{``\texttt{)}''}$}
\State \Return $-1$
\EndIf
\State \Return $\Call{close}{p + 1}$
\EndFunction

\Statex

\Function{rightchild}{$v$}
\If{$\Call{access}{v + 1} = \text{``\texttt{)}''}$}
\State \Return $-1$
\EndIf
\State \Return $\Call{close}{v + 1}$
\EndFunction

\Statex

\Function{isleaf}{$v$}
\State \Return $\Call{access}{v - 1} = \text{``\texttt{(}''}$ and $\Call{access}{v + 1} = \text{``\texttt{)}''}$
\EndFunction

\Statex

\Function{childlabel}{$v$}
\State $p \gets \Call{open}{v}$
\If{$p = 1$}
\State \Return ``root''
\ElsIf{$\Call{access}{p - 1} = \text{``\texttt{(}''}$}
\State \Return ``left''
\Else
\State \Return ``right''
\EndIf
\EndFunction

\Statex

\Function{leftmostdesc}{$v$}
\State $p \gets \Call{open}{v}$
\While{$\Call{access}{p} = \text{``\texttt{(}''}$}
\State $p \gets p + 1$
\EndWhile
\State \Return $p$
\EndFunction

\Statex

\Function{rightmostdesc}{$v$}
\State \Return $\enclose(v) - 1$
\EndFunction

\Statex

\Function{subtreesize}{$v$}
\State \Return $(\enclose(v) - \Call{open}{v}) / 2$
\EndFunction

\Statex

\Function{isancestor}{$u$, $v$}
\State \Return $\Call{open}{u} \le v < \enclose(u)$
\EndFunction

\Statex

\Function{lca}{$u$, $v$}
\If{$u > v$}
\State Swap $u$ and $v$.
\EndIf
\State \Return \Call{rmq}{$u$, $v$}
\EndFunction
\end{algorithmic}
\end{multicols}
\end{algorithm}

Alg.\,\ref{alg:navigation_on_bp} summarizes how to achieve tree navigational operations on the \ac{BP} representation~\cite{Navarro2016Compact}. They return $-1$ if the corresponding node does not exist. The first four functions utilize Prop.\,\ref{prop:order_of_paren}. The next six functions from \textproc{root} to \textproc{childlabel} can be derived naturally from the definition of the BP representation. The next four functions from \textproc{leftmostdesc} to \textproc{isancestor} can be verified by using the fact that the subtree rooted at $v$ corresponds to a half-open interval $[\Call{open}{v}, \enclose(v) )$ on the \ac{BP} representation. The reason the implementation of \textproc{lca} works is as follows. The \ac{LCA} can be characterized as the only common ancestor whose inorder is between the inorders of $u$ and $v$. Also, $w \coloneqq \Call{rmq}{u, v}$ is indeed such a common ancestor; the inorder of $w$ is between the inorders of $u$ and $v$, and the interval corresponding to the subtree rooted at $w$ contains both $u$ and $v$. Alternatively, Ferrada and Navarro~\cite{Ferrada2017ImprovedRMQ} proved the correctness of \textproc{lca} by leveraging the bijection between binary trees and ordinal trees~\cite{Munro2001Succinct}. 

The implementation of $\Call{leftmostdesc}{v}$ uses a simple scan and the execution time depends on the depth from $v$ to the leftmost descendant. Replacing the implementation with $\rankc{}$ and $\selectc{}$ operations makes the execution time independent of depth.

\subsection{The Farzan--Munro Algorithm}
\label{subsec:Farzan_Munro}

The Farzan--Munro tree-covering algorithm~\cite{Farzan2014UniformParadigm} decomposes an ordered tree into smaller-sized trees, called \textit{micro trees}. This decomposition is useful when designing succinct data structures involving trees~\cite{ChakrabortyJSS21,Chakraborty2021CliqueWidth,Davoodi2014EncodingRangeMinima,Farzan2014UniformParadigm,He2014Framework,Munro2021Hypersuccinct,Tsur2018representation}.

The algorithm builds micro trees bottom-up by \ac{DFS} and runs in linear time. A positive integer parameter $B$ is used in the algorithm, which determines the approximate size of the micro trees. Fig.\,\ref{fig:treecover1} shows an example of the tree cover of an ordinal tree. Also, the properties of tree covers are summarized in Prop.\,\ref{prop:originalTC}.

\begin{figure}[hbtp]
\centering
\includegraphics[width=0.7\linewidth]{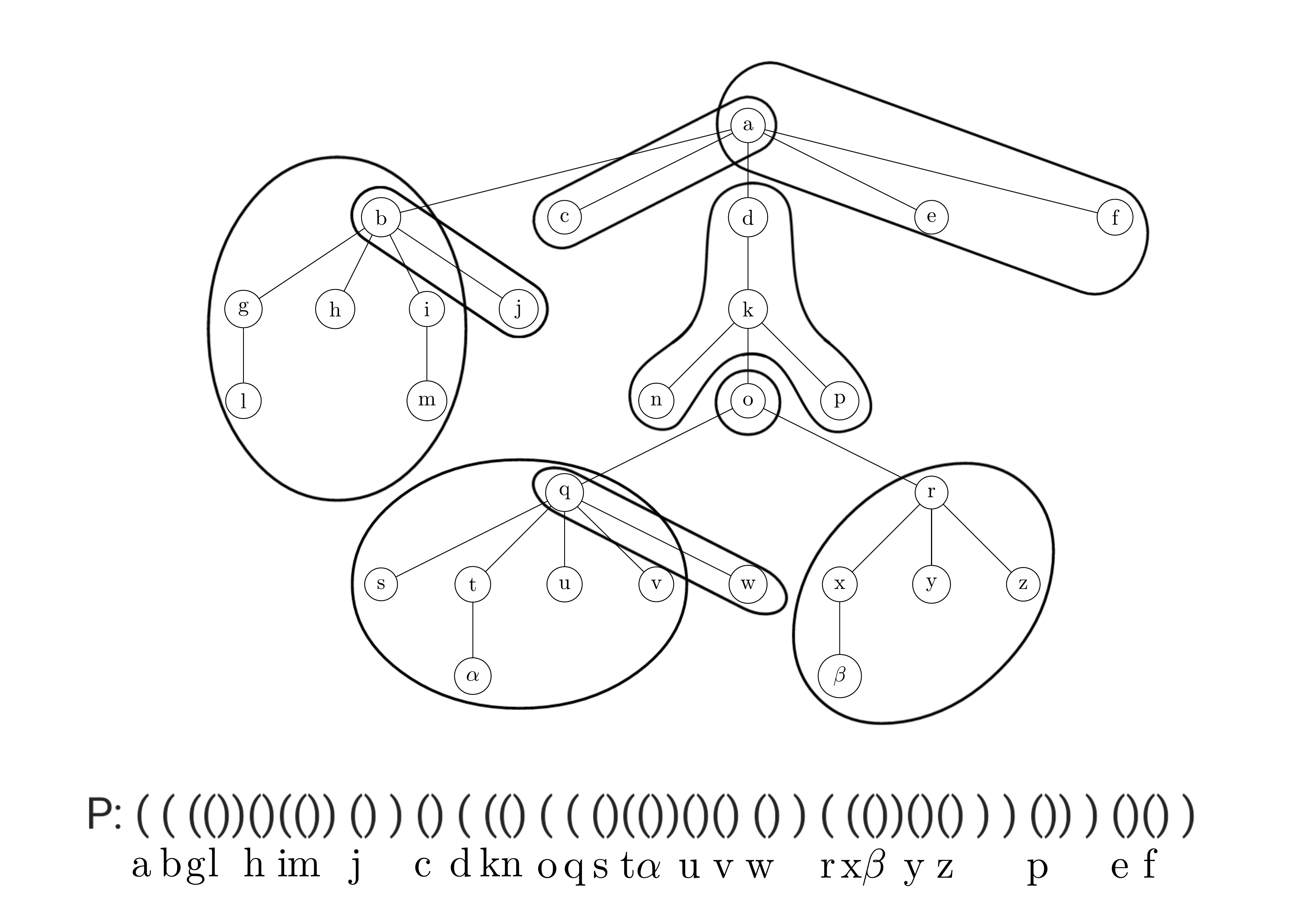}
\caption{The tree cover of a tree, with parameter $B = 5$.}
\label{fig:treecover1}
\end{figure}

\begin{proposition}
[Properties of Tree Covers \cite{Farzan2014UniformParadigm}]\label{prop:originalTC}
For an ordinal tree with $n$ nodes and a parameter $B \ge 1$, the Farzan--Munro algorithm produces a tree cover satisfying the following requirements.
\begin{itemize}
    \item The number of micro trees is $\Order(n/B)$.
    \item Each micro tree has less than $2B$ nodes.\footnote{The original paper~\cite{Farzan2014UniformParadigm} states that the micro-tree size is at most $2B$, but the algorithm does not yield a component of size $2B$.}
    \item Each micro tree has at most one outgoing edge, apart from those from the root of the micro tree.
\end{itemize}
\end{proposition}

Theoretically, the parameter $B$ is often set to $\lceil(\lg n) / 8\rceil$. This is because the number of possible micro trees can be bounded by $\Order(\sqrt{n})$ and look-up tables regarding micro trees consume negligible space. In this paper, we assume $B$ is super-constant with respect to $n$.

When the Farzan--Munro algorithm is applied to binary trees, the micro trees become disjoint. Fig.\,\ref{fig:tree_covering_15} shows an example of the tree cover of a binary tree. The tree cover of a binary tree satisfies Prop.\,\ref{prop:binary_tree_decomposition}.

\begin{figure}[hbtp]
    \centering
    \includegraphics[width=0.5\linewidth]{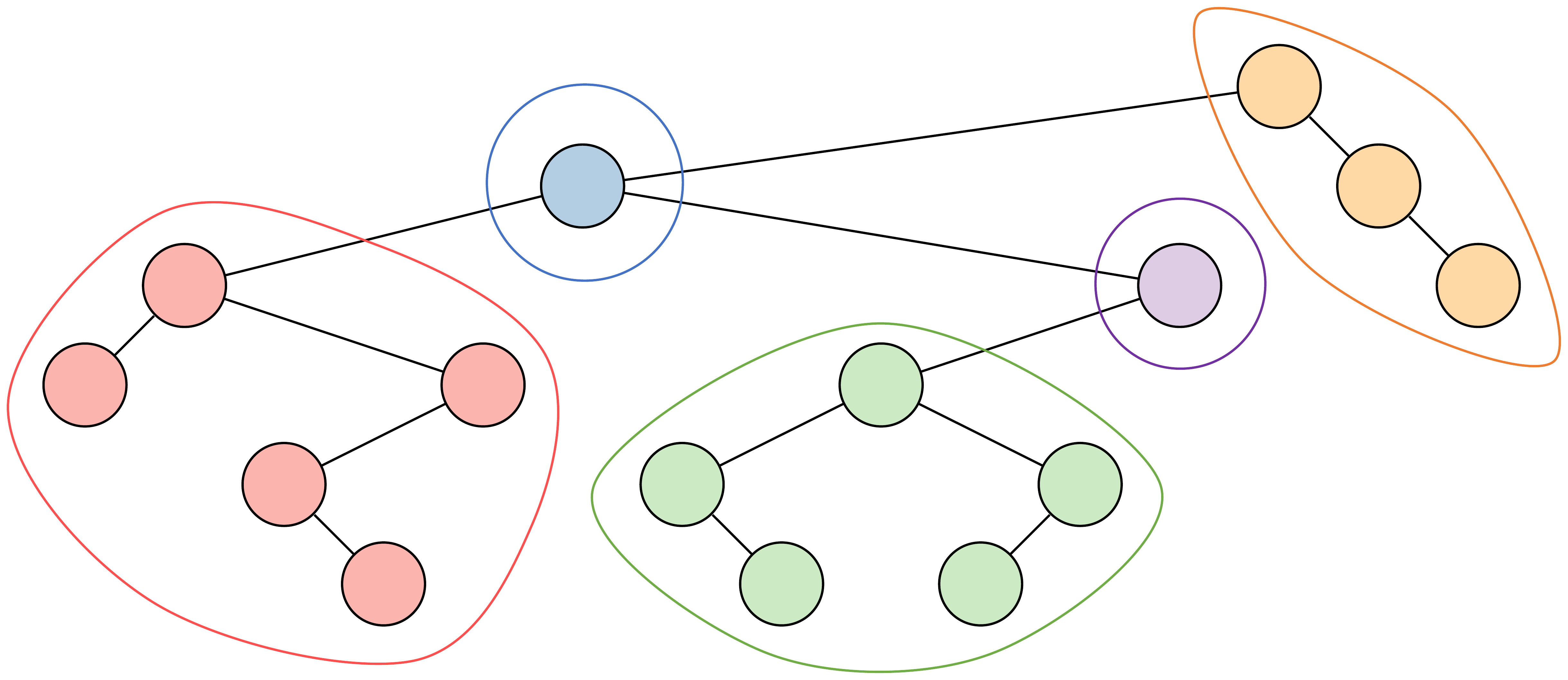}
    \caption{An example of the tree cover of a binary tree with $B = 4$.}
    \label{fig:tree_covering_15}
\end{figure}

\begin{proposition}[Properties of Binary Tree Decomposition~\cite{Farzan2014UniformParadigm, Munro2021Hypersuccinct}]
    \label{prop:binary_tree_decomposition}
    The binary tree decomposition obtained by the Farzan--Munro algorithm satisfies the following additional properties:
    \begin{enumerate}
        \item Contracting the micro trees into single nodes gives a binary tree, which we call \textit{a top-tier tree}.
        \item If a micro tree has two outgoing edges toward children, it consists of a single node.
    \end{enumerate}
\end{proposition}

\subsection{Hypersuccinct Trees}
\label{subsec:hypersuccinct_trees}

Hypersuccinct trees~\cite{Munro2021Hypersuccinct} asymptotically achieve average-case optimal compression for trees from various sources. Here, we present some of their results in binary trees as we later implement the data structure and benchmark it. Applying it to Cartesian trees of random permutations requires $1.736n + \order(n)$ bits on average, which is asymptotically optimal. Also, Cartesian trees of arrays of length $n$ with $r$ increasing runs can be encoded using $2 \lg \binom{n}{r} + \order(n)$ bits. Application of these results to the RMQ problem gives average-case optimal RMQ data structures for those input arrays.

Hypersuccinct trees first utilize the Farzan--Munro algorithm~\cite{Farzan2014UniformParadigm} described in Sec.\,\ref{subsec:Farzan_Munro} to decompose a tree into micro trees, which are subsequently compressed with a Huffman code. Since the micro trees are the most dominant part of the space, compression with a Huffman code reduces the leading term in space complexity for various tree sources.

To explain the data structures of hypersuccinct binary trees, we define portals. To recover the original tree from the decomposed micro trees, it is insufficient just to store the micro trees and the top-tier tree; the connection between nodes of different micro trees cannot be recovered. Thus, we also store where the roots of the child micro trees are initially placed, called \textit{portals}. Each micro tree may contain at most two portals; one is to the left-child micro tree and the other is to the right-child micro tree. Fig.\,\ref{fig:portals} shows an example of possible portal positions. The number of possible portal positions is one more than the number of nodes. The index of the portal from left to right is called the portal rank.

\begin{figure}[tb]
   \centering
  \includegraphics[width=0.30\linewidth]{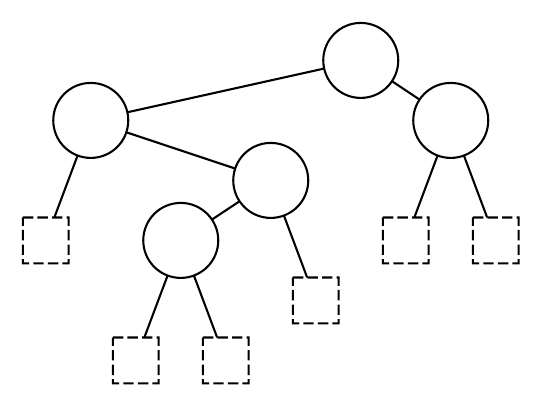}
   \caption{Possible portal positions in a micro tree, which are indicated by dashed squares.}
   \label{fig:portals}
\end{figure}

Thus, a hypersuccinct binary tree consists of micro trees, a top-tier tree, and portals. We consider the space consumption of these items when applied to Cartesian trees of random permutations. The parameter $B$ is set to $\lceil (\lg n) / 8 \rceil$ to make the data structure succinct. 

The micro trees are dominant in space consumption: Huffman codes spend $1.736n + \Order(n (\log B) / B)$ bits, while enumerating the BP representation of each code in a table consumes $\Order(2^{4B} B)$ bits, i.e., $\Order(\sqrt{n} \log n)$ bits.
The top-tier tree and the portals both consume $\order(n)$ bits. The top-tier tree has $\Order(n / B)$ nodes by Prop.\,\ref{prop:originalTC}, so representing the tree by the BP representation requires $\Order(n / B)$ bits. As for the portals, since each micro tree has less than $2B$ nodes by Prop.\,\ref{prop:originalTC}, it has at most $2B$ possible portal positions. Thus, each portal can be represented in $\lceil \lg 2B \rceil$ bits. Since the number of portals is twice the number of micro trees, the portals use $\Order(n (\log B) / B)$ bits.

For later discussion of implementation, we present a sketch of the proof by Munro et al.~\cite{Munro2021Hypersuccinct} that a Huffman code achieves optimal space consumption on average for random permutations. They first introduce a prefix-free code called a \textit{depth-first arithmetic code}. The code encodes a micro tree by listing the left subtree sizes in the \ac{DFS} order and encoding them with arithmetic coding. They show that the use of the depth-first arithmetic code in hypersuccint binary trees achieves asymptotically optimal compression. Then, since a Huffman code is optimal among all the prefix-free codes, space consumption of hypersuccinct binary trees remains optimal if the depth-first arithmetic code is replaced with a Huffman code.

While the depth-first arithmetic code also achieves optimal space consumption, they did not use the depth-first arithmetic code in the data-structure design. This is not surprising since the depth-first arithmetic code takes a longer time to decode and lacks universality for tree sources.
However, we employ the depth-first arithmetic code later in Sec.\,\ref{subsubsec:depth_first_arithmetic_code}, since using the code dispenses with tables for decoding, which critically reduces space consumption in practice.

\section{BP Representation of Tree Covering}
\label{sec:bp_repr_of_tc}

\subsection{Motivation}

As discussed in Sec.\,\ref{subsec:Farzan_Munro}, the Farzan--Munro tree-covering algorithm decomposes a tree into micro trees with desirable properties. Since the decomposition helps speed up the time complexity of queries and reduce space complexity, many succinct data structures involving trees employ this technique~\cite{ChakrabortyJSS21,Chakraborty2021CliqueWidth,Davoodi2014EncodingRangeMinima,Farzan2014UniformParadigm,He2014Framework,Munro2021Hypersuccinct,Tsur2018representation}.

However, it is difficult to simply and compactly represent micro trees and their connection in the succinct data structures based on tree covering. The original representation by Farzan and Munro~\cite{Farzan2014UniformParadigm} employs many pointers using $\order(n)$ bits, but it will practically consume enormous space in total.

In this section, we propose a simple representation of tree covering that leverages the BP sequence. At the core of the representation is that every micro tree corresponds to at most two intervals on the BP representation. This property motivates us to use the index of the BP sequence while preserving the tree-covering structure. We then give data-structure designs that isolate micro trees as a sequence and enable compressing micro trees with an arbitrary encoding. 

In the following sections, we first follow the above argument for ordinal trees by modifying the tree covers. We then show that a similar argument holds for binary trees. We also improve space consumption by exploiting properties specific to binary trees. 

\subsection{Ordinal Trees}
\label{subsec:ordinal_trees}

\subsubsection{Modified tree cover}
\label{subsubsec:modified_tree_cover}

We slightly modify the definition of the tree cover
and give its simple representation using multi-type parentheses sequences.
We then show how to obtain the original tree cover from the modified tree cover.

First, we compute the original tree cover of an ordinal tree $t$, using the tree cover algorithm of Farzan and Munro~\cite{Farzan2014UniformParadigm}.
Let $v$ be a root of a micro tree $\mu$ consisting of multiple nodes. Then, we create a dummy node $w$ and hang it from $v$.
We change the parent of children of $v$ in $\mu$ as $w$. 
We call the resulting tree $t'$. The modified tree cover is obtained by splitting $\mu$ into two micro trees: a singleton micro tree that contains only $v$ and a micro tree that contains $w$ and its children in $\mu$.

The modified tree cover obtained by the procedure above consists of disjoint micro trees. As in the case of binary trees, we define a \textit{top-tier tree} as a tree obtained by contracting the micro trees of the modified tree cover.
The theorem below describes the properties of the modified tree cover.
\begin{theorem}\label{th:newTC}
For an ordinal tree with $n$ nodes and a parameter $B \ge 1$, we can obtain a modified tree cover satisfying the following requirements.
\begin{enumerate}
    \item The number of micro trees is $\Order(n/B)$.
    \item Each micro tree either consists of a single node with an arbitrary number of child micro trees (called a singleton micro tree), or it consists of less than $2B$ nodes and has at most one child micro tree (called a non-singleton micro tree).
\end{enumerate}
\end{theorem}

\begin{figure}[tb]
\centering

\begin{subfigure}{.55\textwidth}
\includegraphics[width=1\linewidth]{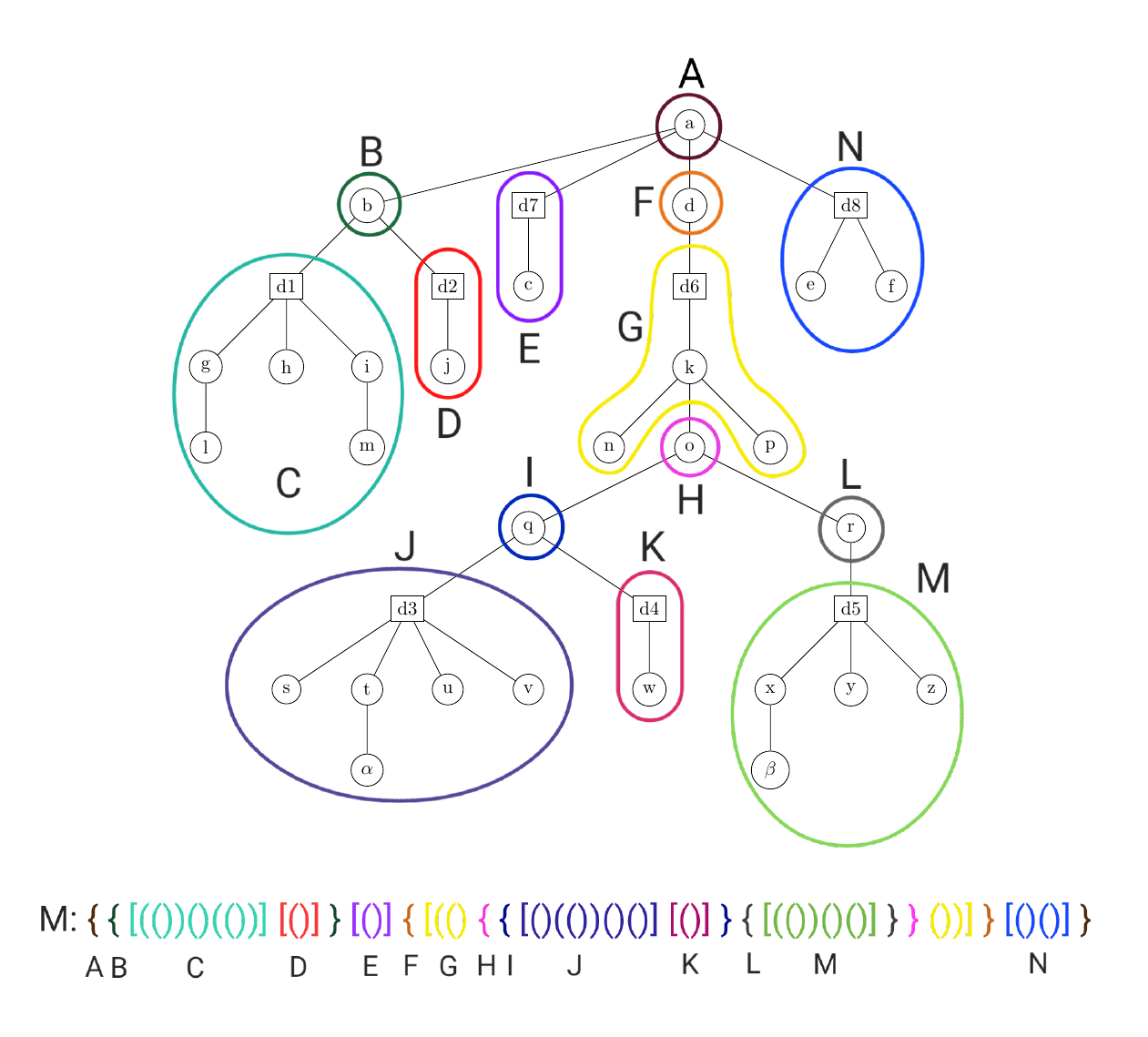}
\caption{Modified tree cover.}
\label{fig:treecover1a}
\end{subfigure}
\hspace{0.02\textwidth}
\begin{subfigure}{.4\textwidth}
\includegraphics[width=1\linewidth]{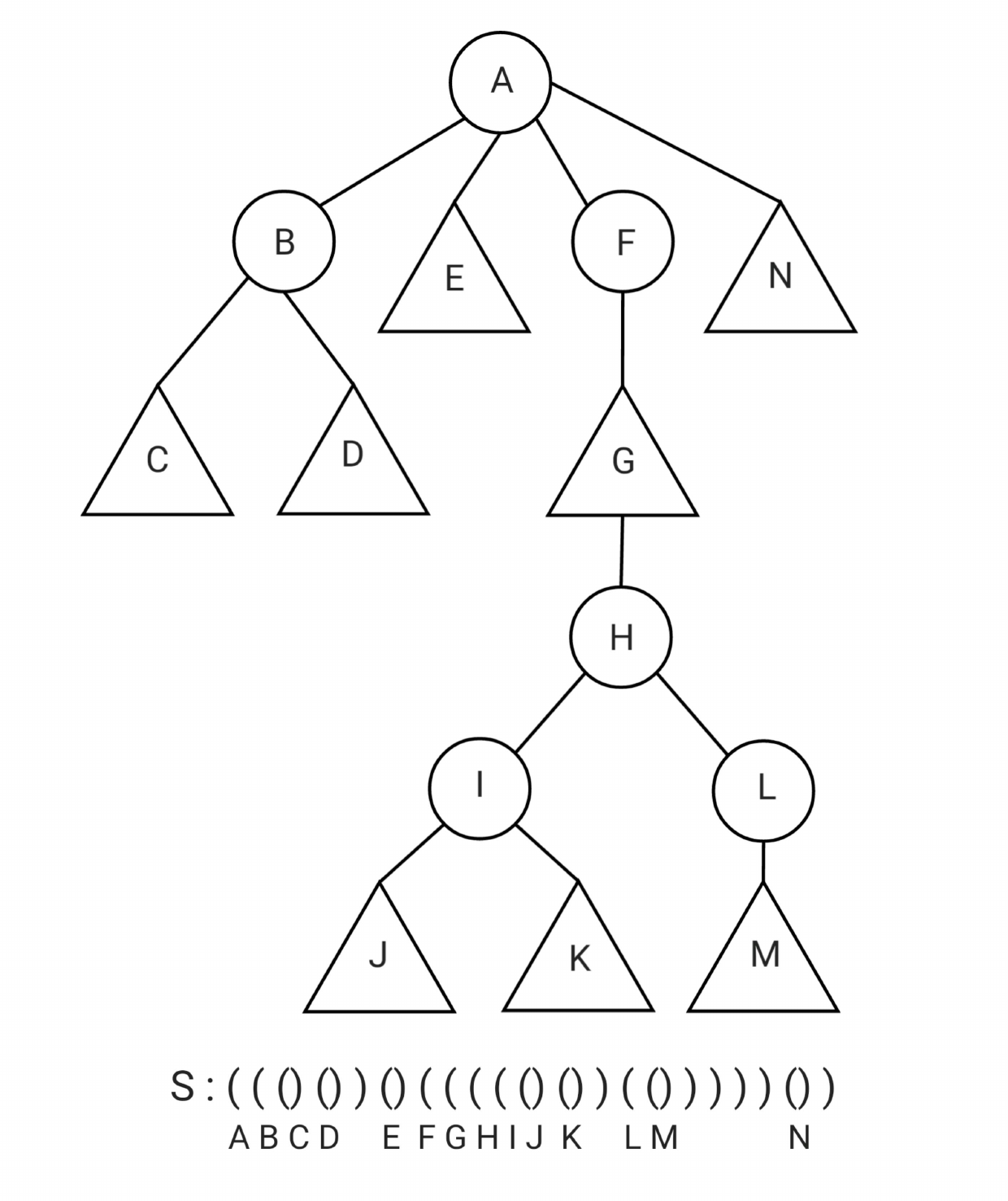}
\caption{Top-tier tree.}
\label{fig:treecover1b}
\end{subfigure}
\caption{The modified tree cover and the top-tier tree of the tree in Fig.\,\ref{fig:treecover1}.}
\end{figure}

We show how to represent the modified tree cover of a modified tree $t'$ using multi-type parentheses sequence.
We use three types: {\tt ()}, {\tt \{\}}, and {\tt []}.
Normal parentheses {\tt ()} represent non-root nodes of micro trees.
Curly braces {\tt \{\}} represent root nodes of singleton micro trees.
Square brackets {\tt []} represent root nodes of non-singleton micro trees, i.e., the dummy nodes inserted in $t$.
We simply construct a multi-type parentheses sequence $M$ during a \ac{DFS} of $t'$. 

This multi-type parentheses sequence $M$ has the following good properties.
\begin{enumerate}
    \item If we remove all normal parentheses from $M$ and convert curly braces and square brackets into normal parentheses, it coincides with the BP sequence $S$ of the top-tier tree.
    \item If we remove all square brackets from $M$ and convert curly braces into normal parentheses, it coincides with the BP sequence $P$ of $t$.
    \item For any micro tree $\mu$, its BP representation is cut into at most two parts. 
\end{enumerate}

Let us briefly explain why the third property holds. There are two possible cases why a micro-tree BP splits: an outgoing edge emanates from a non-root node or another micro tree hangs to the micro-tree root and is between two children of the micro-tree root. The former case occurs at most once by Prop.\,\ref{prop:originalTC}. Due to the greedy strategy of the Farzan--Munro algorithm, the latter case rarely happens: it can be verified by examining the detail of the algorithm that the latter case occurs only when the micro-tree root has only one heavy child, which then implies that the micro tree has no other outgoing edge. Thus, the latter case occurs at most once and both the former and latter cases do not simultaneously happen.

The third property also holds for the original BP sequence $P$ if we consider the partition of $P$ into micro-tree BP sequences and apply the conversion described in the second property. This is equivalent to excluding the micro-tree root when considering the BP representation of a non-singleton micro tree. We assume this implicitly when we consider a micro-tree BP of $P$. Otherwise, a micro tree can correspond to more than two intervals of $P$. Also, for simplicity, we arbitrarily split a single interval corresponding to a whole micro tree into two intervals so that every micro tree corresponds to two intervals; we call the resulting two intervals \textit{chunks}.

We observe that any micro tree in the modified tree cover either has a dummy node as its root or is a singleton micro tree. Let $\mu$ be a non-singleton micro tree in the modified tree cover. The original tree cover is obtained by simply merging $\mu$ with the parent node of the dummy root node of $\mu$.  

Note that the sequence $M$ leads to a succinct representation of the tree $t$ and its tree cover.
Because the length of $M$ is $2n+\Order(n/B)$ and the number of curly braces and square brackets is $\Order(n/B)$, we can store $M$ in $2n+\Order(n (\log B) / B + n (\log \log n) / \log n)$ bits by using sparse bitvectors~\cite{RamanRS07}.

\subsubsection{Practical Designs}
Although the sequence $M$ leads to an intuitive and succinct representation of ordinal trees, there is still room for simplification in practice. In what follows, Theorem~\ref{thm:TC_index} first presents a practical design of tree-covering indexes for ordinal trees.  Theorem~\ref{thm:replace_with_rmm} then discusses the replacement of the indexes to handle tree navigational queries. Theorem~\ref{thm:TC_index_modified} finally describes another design that enables compressing micro trees with an arbitrary encoding, which is suitable for designing succinct representations for subclasses of ordinal trees.

We first describe the role of curly braces and square brackets when viewing $M$ as a representation of the original tree cover:
\begin{enumerate}
    \item Curly braces and square brackets together delimit the original BP sequence $P$ based on the original tree cover.
    \item The difference between curly braces and square brackets is whether the micro-tree BP is complete or not; a micro-tree BP enclosed with square brackets lacks the micro-tree root and needs to be enclosed with a pair of parentheses to obtain the complete micro-tree BP.
\end{enumerate}

Thus, instead of inserting curly braces and square brackets into $P$, it suffices to store where to delimit the sequence $P$ and whether each micro-tree BP is complete or not. Therefore, we present a practical succinct representation of ordinal trees based on this approach.
\begin{theorem}\label{thm:TC_index}
Let $\mu_1, \dots, \mu_m$ be the micro trees in the \ac{DFS} order on the top-tier tree.
The following indexes consume $\order(n)$ bits in total and represent the tree cover if combined with the BP sequence $P$.
\begin{itemize}
    \item A sparse bitvector $V$ of the same length as $P$ which marks the starting position of each chunk of $P$ with $1$. It handles rank and select operations in constant time and spends $\Order(n (\log B) / B + n (\log \log n) / \log n)$ bits~\cite{RamanRS07}.
    \item A data structure that represents the BP sequence $S$ of the top-tier tree. It handles basic operations on the BP sequence $S$ in constant time and spends $\Order(n / B)$ bits~\cite{Navarro2014FullyFunctional}.
    \item A boolean array $F$ whose $i$-th element indicates whether the BP of $\mu_i$ in $P$ is complete or not, i.e., whether $\mu_i$ consists of a single node or not. It consumes $\Order(n / B)$ bits.
\end{itemize}
\end{theorem}

Given a position $i$ of $P$, we can determine the micro tree containing $i$ as follows. First, the index $r$ of the chunk can be obtained by $\rank_{1}(V, i)$. Then, the range of the chunk on $P$ is from $\select_{1}(V, r)$ to $\select_{1}(V, r + 1)$. To fully restore the micro tree, we also need to retrieve the other chunk that is a part of the same micro tree, and the index of the other chunk can be obtained by using either \open{} or \close{} operation on $S$ and $r$. Finally, we refer to the array $F$ to see if the micro-tree root needs to be merged. The index of $F$ can be obtained by $\rank_{(}$ operation on the index of the left chunk of the micro tree.

One advantage of the indexes in Theorem~\ref{thm:TC_index} is that they are compatible with the indexes of the BP representation, especially with the range min-max tree~\cite{Navarro2014FullyFunctional}. We can replace $V$ with the range min-max tree by modifying the algorithm: instead of fixed-length blocks in the original algorithm, we use variable-length blocks so that the leaves of the range min-max tree correspond to the chunks of $P$. Thereby, the data structure can handle BP operations on $P$ and thus tree navigational queries on $t$.

\begin{theorem}\label{thm:replace_with_rmm}
    Consider the data structures in Theorem~\ref{thm:TC_index}. Tree navigational queries can be supported by replacing the sparse bitvector $V$ with a range min-max tree that uses variable-length blocks.
\end{theorem}

\begin{figure}[tb]
\begin{center}
\includegraphics[scale=0.5]{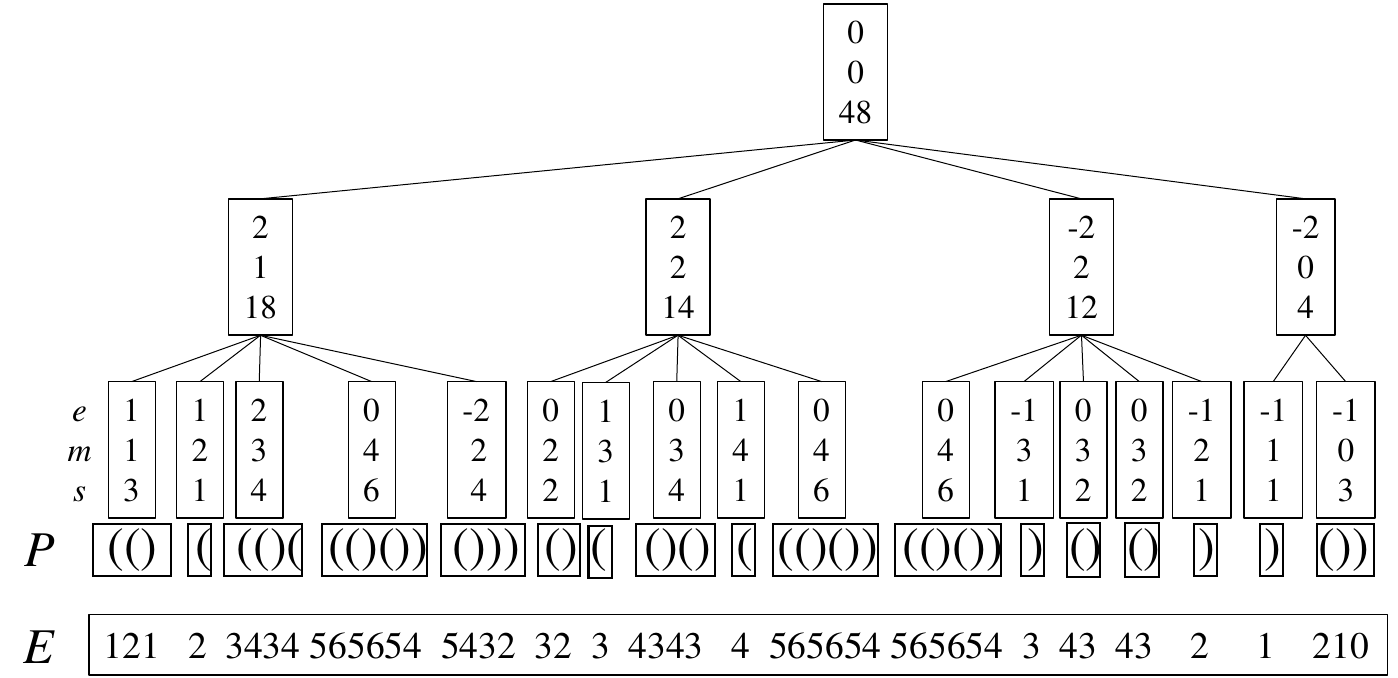}
\caption{The range min tree of a BP sequence.}
\label{fig:treecover2}
\end{center}
\end{figure}

Fig.\,\ref{fig:treecover2} gives an example of the range min tree~\cite{10.1145/2656332}, which is a simplified variant of the range min-max tree. For each chunk, we store the local excess value $e$, that is, the number of opening parentheses minus the number of closing parentheses in the chunk. We also store the minimum excess value $m$ in the chunk.
In addition to these two values, we store the size $s$ of the chunk because chunks have different lengths. The leaves of the range min tree store tuples $(e, m, s)$. An internal node of the range min tree also stores the tuple for the BP sequence made by concatenating those for the children. It is shown~\cite{Navarro2014FullyFunctional} that, if the total length of the BP sequence for all the chunks is $m = {\rm polylog}(n)$, the range min-max tree can be stored in $\Order(m (\log\log n)/\log n)$ bits and any basic tree operation can be done in constant time.  This can be easily extended for the case that chunks are of variable length of $\Order(\log n)$. Since the $\rank$ and $\select$ operations on $V$ can be simulated by using $s$ on the range min-max tree, we can retrieve micro trees as in Theorem~\ref{thm:TC_index}.

Another advantage of the indexes in Theorem~\ref{thm:TC_index} is that micro trees can be compressed with an arbitrary encoding by the following modified scheme. 

\begin{theorem}\label{thm:TC_index_modified}
Let $\mu_1, \dots, \mu_m$ be the original micro trees in the \ac{DFS} order on the top-tier tree. The function $D(\mu_i)$ represents the code of the micro tree $\mu_i$. 

Suppose we already have the indexes $V$ and $S$ described in Theorem \ref{thm:TC_index}. The sparse bitvector $V$ can be replaced with the range min-max tree to support tree navigational queries as in Theorem~\ref{thm:replace_with_rmm}.
Then, the storage of $P$ can be replaced with a variable-length cell array of codes $D(\mu_1), \dots, D(\mu_m)$ supporting random access.
Note that $F$ is unnecessary since it can be obtained by checking if $\mu_i$ consists of a single node or not. 
\end{theorem}

We can recover $P$ from the indexes: for every micro tree, we can identify where it is placed on $P$ as in Theorem~\ref{thm:TC_index}, and we can retrieve the micro-tree BP by accessing the corresponding code in a similar way as accessing $F$ in Theorem~\ref{thm:TC_index}. Also, random access on $P$ can be efficiently supported by this method.

One application of Theorem~\ref{thm:TC_index_modified} is optimal compression of unordered trees. The optimal compression of unordered trees needs to fit into $1.564 n + \order(n)$ bits~\cite{Otter1948NumberOfTrees}. Farzan and Munro~\cite{Farzan2014UniformParadigm} achieve this by creating a look-up table $T_k$ that lists all the possible micro-tree shapes with $k$ nodes and encoding each micro tree into its size $s$ and its index $o$ of the table $T_s$. This technique is compatible with the scheme in Theorem~\ref{thm:TC_index_modified}, giving another succinct representation of unordered trees with fewer indexes.

Another application is hypersuccinct trees~\cite{Munro2021Hypersuccinct}, which achieves optimal compression for various tree distributions by encoding micro trees with a Huffman code. Using a Huffman code in Theorem~\ref{thm:TC_index_modified}, we obtain a practical design of hypersuccinct trees.

\subsection{Binary Trees}

\subsubsection{Extension from Ordinal Trees to Binary Trees}

There are several differences to consider when applying the BP representation of tree covering to binary trees. First, the micro trees do not share nodes when the Farzan--Munro algorithm is applied to binary trees. Thus, we do not need to modify tree covering as in Sec.\,\ref{subsubsec:modified_tree_cover}.
Also, as introduced in Sec.\,\ref{subsubsec:bp_def}, the BP representation of a binary tree differs from that of ordinal trees. Thus, how the tree-covering hierarchy appears in the BP sequence should be again investigated. 

We note the advantages of using the BP representation for binary trees. First, since the representation distinguishes a single left child and a single right child, we do not need to tailor our scheme to full binary trees. Also, we can directly handle queries involving inorder values on the BP sequence, which is especially beneficial when applying to Cartesian trees to implement average-case optimal \ac{RMQ} data structures.

In particular, we expect RMQ data structures using hypersuccinct binary trees~\cite{Munro2021Hypersuccinct} to be more efficient in practice. While the RMQ data structure achieves optimal space consumption for random permutations and processes queries in constant time, there are some problems in practice. First, it handle queries involving inorder by reduction to several queries~\cite{Davoodi2014EncodingRangeMinima}, which will make \ac{RMQ} implementation inefficient in both time and space. Also, it requires many $\order(n)$-bit indexes and will consume a significant amount of space. Thus, we need to incorporate the tree-covering structure into the BP sequence of binary trees so that the data structure becomes more efficient in both time and space.

In the remainder of this section, we first show that every micro tree corresponds to at most two intervals on the BP representation, which enables adopting the data structures in Theorems~\ref{thm:TC_index}--\ref{thm:TC_index_modified} to binary trees. We then introduce a \textit{split rank}, which corresponds to portals but is more versatile and compact. Finally, we discuss the data-structure design utilizing the observations.

\subsubsection{Top-Tier Tree}
Here, we discuss that each micro-tree BP appears as one or two intervals in the BP representation of the original tree, and replacing each one or two intervals with a pair of parentheses gives the BP of the top-tier tree. Fig.\,\ref{fig:top_tier_schematic} summarizes the overview.

\begin{figure}[tb]
    \centering
    \includegraphics[width=0.9\columnwidth]{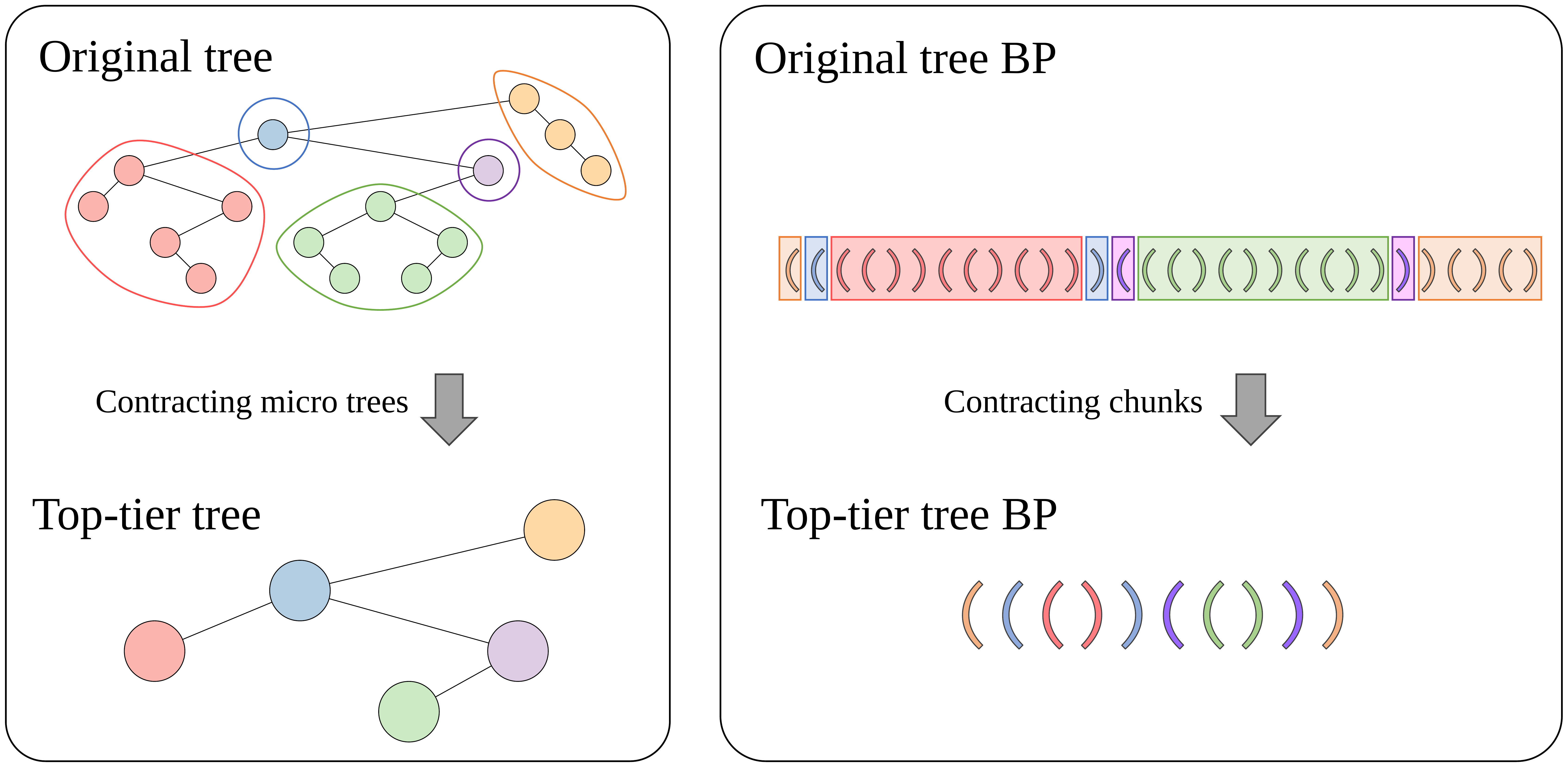}
    \caption{Overview of the relationship between the original tree and the top-tier tree. Each micro tree corresponds to one or two intervals of the BP sequence. Just as contracting micro trees yields the top-tier tree, so contracting chunks into parentheses yields the BP representation of the top-tier tree.}
    \label{fig:top_tier_schematic}
\end{figure}

\begin{proposition}
    \label{prop:binary_chunks_at_most_two}
    Each micro tree corresponds to one or two intervals of the BP representation of the original tree.
\end{proposition}
   
\begin{proof}
    Since a subtree corresponds to an interval of the BP, the intervals of a micro tree can be obtained by first considering an interval that corresponds to a subtree rooted at the micro-tree root and then removing sub-intervals that correspond to subtrees rooted at the roots of left-child and right-child micro trees. If the number of the child micro trees is at most one, the number of the intervals is at most two, thus satisfying the proposition; if the number of the child micro trees is two, then Prop.\,\ref{prop:binary_tree_decomposition} implies that the micro tree consists of a single node and its BP consists of two parentheses, which immediately shows the proposition.
\end{proof}

For simplicity, if a micro tree corresponds to a single interval, we virtually split it into two intervals so that every micro tree corresponds to two intervals, which we call \textit{chunks}.
\begin{definition}
    \label{def:chunks}
    A single interval representing a whole micro tree is split into two intervals: a left interval is defined to be the original interval, followed by a right interval of zero width. Thus, every micro tree corresponds to two intervals, which we call \textit{chunks}.
\end{definition}

The reason for this extreme division is that it works well with a variant of the top-tier tree and split ranks, which are discussed later. 
Where to split a single interval does not affect the theorem below.

\begin{theorem}[BP of top-tier tree]
    \label{thm:top_tier_tree_BP}
    Contracting each left and right chunk into a left and right parenthesis yields the BP of the top-tier tree.
\end{theorem}
\begin{proof}    
    It suffices to show that the edges of the top-tier tree also appear in the contracted BP. Note that the BP representation has an edge between a parent $p$ and a child $c$ if and only if the opening parenthesis corresponding to $c$ is immediately right next to either of the matching parentheses corresponding to $p$. 
    
    We take an arbitrary edge of the top-tier tree. Let $\mu_{p}$ and $\mu_{c}$ be the parent micro tree and the child micro tree that are connected by the edge. Then, one node of $\mu_p$ and one node of $\mu_c$ are connected in the original tree; let $p$ and $c$ be the connected node from $\mu_p$ and $\mu_c$, respectively. Note that $c$ must be the root of $\mu_c$. Then, in the BP representation of the original tree, the opening parenthesis of $c$ is right next to either of the matching parentheses of $p$. Since the opening parenthesis of the micro-tree root $c$ is the first element of $\BPb(\mu_c)$, it belongs to the left chunk. Thus, the left chunk of $\mu_c$ is right next to a chunk of $\mu_p$. After contracting chunks, the left parenthesis of $\mu_c$ is right next to either of the parentheses of $\mu_p$, which indicates that the contracted BP has an edge between the parent $\mu_p$ and the child $\mu_c$.
\end{proof}

Note that the above theorem defines whether a single child in the top-tier tree is a left- or right-child, which is not discussed in the definition of the top-tier tree. It also shows that the original BP can be obtained by properly replacing pairs of matching parentheses in the top-tier tree BP with the corresponding micro-tree BP sequences.

Although not used in the implementation, we discuss a variant of the top-tier tree. This requires splitting a single interval at the right end when defining a chunk: it defines the left chunk as the original interval followed by the right chunk of zero width. Then, if a micro tree has a single right child, its right chunk is empty. Thus, even if we change the single right child to the single left child, it only moves in the top-tier tree BP the parentheses corresponding to zero-width chunks. Therefore, Theorem \ref{thm:top_tier_tree_BP} still holds if we appropriately move the chunks and use a variant of a top-tier tree that defines all the single children as left children.

This variant has the advantage that it may be further compressed: it is a Motzkin tree, which hypersuccinct binary trees encode using $1.585n + \order(n)$ bits~\cite{Munro2021Hypersuccinct}. A drawback of this variant is that it prevents efficient execution of $\enclose{}$ in the data structure, which is discussed in Sec.\,\ref{subsubsec:minimal_design}.

\subsubsection{Split Rank}
Theorem \ref{thm:top_tier_tree_BP} enables recovery of the original BP by appropriately replacing each matching pair of parentheses in the top-tier tree BP with the micro-tree BP.
To achieve the appropriate replacement, we need to remember where to split each micro-tree BP as in Theorem~\ref{thm:TC_index_modified}.
As for binary trees, we find a versatile and compact value that determines where to cut a micro-tree BP, named a \textit{split rank}. It naturally corresponds to portals and is useful for queries involving inorder, while consuming one bit less than naively storing where to split a micro-tree BP. In addition, storing split ranks consumes only half as much space as storing portals, since only one split rank is needed per micro tree, whereas two portals are needed per micro tree. We start with a lemma that enables more efficient storage than naively maintaining where to split the micro-tree BP.

\begin{lemma}
    \label{lem:right_chunk_begins_with_close}
    Non-empty right chunks begin with closing parentheses.
\end{lemma}
\begin{proof}
    Since the micro-tree BP splits into two intervals, there is a child micro-tree BP between the two intervals. By adding the root of the child micro tree as a new leaf to the current micro tree, we can identify where the micro-tree BP can be split. Adding the leaf to the micro tree corresponds to inserting a pair of parentheses between the chunks. If the right chunk begins with an opening parenthesis, the node represented by the opening parenthesis becomes the right child of the node indicated by the inserted pair, contradicting the fact that the inserted pair is a leaf.
\end{proof}
By the above lemma, the number of closing parentheses in the left chunk suffices to recover where to split the micro-tree BP. We define the quantity plus one as a \textit{split rank} so that it corresponds to the portal ranks, as discussed in Prop.\,\ref{prop:split_and_portal}.

\begin{definition}[Split rank]
    The \textit{split rank} of a micro tree is defined to be the number of closing parentheses in the left chunk plus one.
\end{definition}

Calling $\selectc{}$ with the split rank returns the starting index of the right chunk. Assuming that \select{} operation returns the last index plus one when searching for a non-existent key, it also works well when the micro-tree BP appears as a single interval in the original tree BP. Also, the split rank corresponds to the rank of the leftmost portal.
\begin{proposition}
    \label{prop:split_and_portal}
    The split rank equals the rank of the portal to the leftmost child if it exists.
\end{proposition}
\begin{proof}
    If the micro tree has two child micro trees, the micro tree consists of a single node by Prop.\,\ref{prop:binary_tree_decomposition}. Thus, both values equal one and the proposition follows. 
    
    We first consider the case when the micro tree has a single child micro tree and the right chunk is non-empty. Then, we can utilize the same settings as the proof of Lemma~\ref{lem:right_chunk_begins_with_close}: we consider adding the child micro-tree root to the current micro tree as a new leaf, i.e., inserting a pair of parentheses between the chunks. Then, the micro-tree-local rank of the inserted closing parenthesis equals the split rank, thus the new leaf inorder equals the split rank as well. Also, as illustrated in Fig.\,\ref{fig:portals}, the inorder of the new leaf is equal to the portal rank. 

    We then discuss the case when the micro tree has a single child micro tree and the right chunk is empty. Then, the child micro tree needs to be at the last position of the possible portals to avoid the BP of the parent micro tree split in the middle. Thus, both the split rank and the portal rank equal the number of nodes in the micro tree plus one. This finishes the proof.
\end{proof}

\subsubsection{Counterparts of Practical Designs}
In what follows, we consider the counterparts of Theorems~\ref{thm:TC_index}--\ref{thm:TC_index_modified}. 
In the counterpart of Theorem~\ref{thm:TC_index}, we do not need $F$ since micro trees are disjoint. Thus, it just marks the starting points of the chunks with $V$, which is a natural way to store the tree cover. We also note that when $V$ is used to indicate the starting positions, we should not adopt the split positions in Def.\,\ref{def:chunks} since the starting points may collide and $V$ needs to be replaced with a multiset. Here, we can adopt the rule of cutting the intervals at some internal positions so that every chunk has a nonzero width and the start points become distinct.

Theorems~\ref{thm:replace_with_rmm} and \ref{thm:TC_index_modified} directly apply to binary trees.
It is now possible to use zero-width chunks. The data structures can exploit both the BP sequence and the tree-covering structure to efficiently handle tree navigational queries. For example, they support efficient inorder conversion by using the BP sequence, which is conventionally solved by reduction to several queries~\cite{Davoodi2014EncodingRangeMinima} in tree covering. They also enable efficient depth query that is difficult with the $\BPb$ representation but is easy if we use the tree-covering structure.

\subsubsection{Minimal Design}
\label{subsubsec:minimal_design}

Here, we present a more efficient design that leverages the hierarchical structure of the BP sequence and the limited use of BP operations due to tailoring to binary trees.
To support tree navigational queries, Theorem~\ref{thm:TC_index_modified} employs the range min-max tree built on $P$ to support the BP operations on $P$. However, if we limit the BP operations to ones necessary for the tree navigational queries, we can simplify the data structure. In particular, we can still support the tree navigational queries if we remove the minimum excess value $m$ in the nodes of the range min-max tree built on $P$, i.e., if we replace the range min-max tree with prefix sum data structures~\cite{RamanRS07}.

\begin{theorem}
    \label{thm:minimal_design}
    Among the data structures in Theorem~\ref{thm:TC_index_modified}, we can replace the range min-max tree constructed on $P$ with two arrays $O$ and $C$ which stores the number of opening and closing parentheses in each chunk and supports efficient prefix sum queries. More formally, the following data structures are sufficient to support the tree navigational queries in Alg.\,\ref{alg:navigation_on_bp}.
    \begin{itemize}
        \item A variable-length cell array of codes $D(\mu_1), \dots, D(\mu_m)$ supporting random access.
        \item A data structure that represents the BP sequence $S$ of the top-tier tree. It handles basic operations on the BP sequence $S$ in constant time and spends $\Order(n / B)$ bits~\cite{Navarro2014FullyFunctional}.
        \item Two arrays $O$ and $C$ of the same length as the number of chunks, defined as follows: the $i$-th element of $O$ and $C$ stores the number of opening and closing parentheses in the $i$-th chunk, respectively. The arrays $O$ and $C$ support queries involving prefix sum~\cite{RamanRS07}.
    \end{itemize}
\end{theorem}

There are two main roles of $O$ and $C$. One is to convert a node, its preorder, and its inorder to one another. The other is to indicate the length of each chunk, which can be obtained by adding the elements of $O$ and $C$.
Note that the values of $e$ and $s$ in each chunk in the range min-max tree can be easily recovered from the values of $O$ and $C$.

\begin{remark}
    \label{rem:actual_impl}
    We mainly use the above mechanism in the later implementation, but we modify the mechanism to further reduce the space consumption: it removes the array $O$ and replaces the array $C$ with the array of split ranks in the \ac{DFS} order of the top-tier tree and sparse sampling of the elements of $C$.
    By removing $O$, we need to give up converting a node and its preorder to each other. On the other hand, since we can find where a micro-tree BP splits on the BP sequence $P$ by using the number of micro-tree nodes and the split rank, we can still recover the length of each chunk.
    Also, the elements of $C$ can be efficiently recovered from the array of split ranks. In particular, the number of closing parentheses in the left chunk equals its split rank minus one by definition, and that in the right chunk can be obtained by subtracting that of the left chunk from the number of micro-tree nodes. Thus, we can support the conversion of a node and its inorder to each other.
\end{remark}

In this minimal design described in Theorem~\ref{thm:minimal_design}, each parenthesis is represented as a pair of indices: an index of a chunk and a chunk-local index. In what follows, we show that the BP operations given in Sec.\,\ref{subsec:balanced_parentheses} can be efficiently executed.

We first discuss the possible operations by using only the main three components. The following queries can be efficiently supported: \access{}, \open{}, \close{}, and $\enclose{}$. Using queries on the top-tier tree BP, we can convert the index of the micro-tree and the index of the chunk to each other. Thus, we can retrieve the micro-tree BP, which enables \access{}, \open{}, and \close{}. It also enables $\enclose{}$: the only difficult case is when the operation needs to search outside the micro-tree BP, but in that case, the enclosing parenthesis belongs to the right chunk that encloses the micro-tree BP, which can be retrieved by the $\enclose{}$ operation on the top-tier tree BP. 

Although we can also implement the \textproc{rmq} operation by using \textproc{rmq} on the top-tier tree, it is much more straightforward to directly implement \textproc{lca} by using the tree-covering property. The following proposition presents the steps needed to calculate \ac{LCA}.

\begin{proposition}
    \label{prop:lca_in_minimal}
    The following steps show how to calculate $\Call{lca}{u, v}$.
    \begin{enumerate}
        \item If $u$ and $v$ belong to the same micro tree, we calculate the \ac{LCA} within the micro tree and return the corresponding pair of indices.
        \item Otherwise, we compute the \ac{LCA} micro tree of the microtrees to which $u$ and $v$ belong.
        \item If the LCA micro tree is different from the micro trees to which $u$ and $v$ belong, then the LCA micro tree consists of a single node by Prop.\,\ref{prop:binary_tree_decomposition}. Thus, we can identify the corresponding node and return its pair of indices. 
        \item Otherwise, we can assume that the micro tree of $u$ coincides with the LCA micro tree without loss of generality. Then, we take the node that emits the portal to the subtree containing $v$, calculate the LCA of the taken node and $u$, and return its pair of indices. We can utilize Prop.\,\ref{prop:split_and_portal} to find the node that emits the portal.
    \end{enumerate}
\end{proposition}

Thus, the remaining operations are the \rank{} and \select{} operations. For these operations, we first scan chunk-wise by using the arrays $O$ and $C$. We then inspect the chunk by retrieving the micro-tree BP sequence.

Using the BP operations discussed above, we describe how to implement tree navigational operations listed in Alg.\,\ref{alg:navigation_on_bp}. Since we can calculate the pair of indices that expresses an adjacent parenthesis, we can simulate the statements $v \pm 1$ in Alg.\,\ref{alg:navigation_on_bp}. Thus, by using the BP operations, most of the operations can be naturally implemented in our mechanism as well.

In what follows, we discuss some exceptional operations that cannot be implemented naturally. First, \textproc{subtreesize} subtracts between two indices to count the number of parentheses in the subtree, which cannot be implemented in our mechanism. We can instead use the formula $\Call{subtreesize}{v} = \Call{rightmostdesc}{v} - \Call{leftmostdesc}{v} + 1$ to avoid subtraction between indices. Also, the implementation of \textproc{isancestor} needs some consideration. To check if a pair of indices is between two pairs of indices, we need to compare two pairs of indices lexicographically: the pair with the larger chunk index is defined to be larger; if the chunk indices are equal, the pair with the larger chunk-local index is defined to be larger; otherwise, the pairs are equal.

\subsubsection{Design Comparison}

Here, we compare the data-structure designs in Theorems~\ref{thm:TC_index_modified} and \ref{thm:minimal_design}. The advantage of the data structure in Theorem~\ref{thm:minimal_design} is that it is space-efficient compared to the other design. This is especially crucial when we use a Huffman code to encode micro trees and a small value of $B$, as we later discuss in Sec.\,\ref{sec:space_opt_motivation}. 

The advantage of the data structure in Theorem~\ref{thm:TC_index_modified} is that it can easily support some extra BP operations that are difficult with the data structure in Theorem~\ref{thm:minimal_design}. For example, the data structure in Theorem~\ref{thm:TC_index_modified} supports \textproc{fwdsearch} and \textproc{bwdsearch}, the generalized operations of \open{}, \close{}, and $\enclose{}$~\cite{Navarro2016Compact, Navarro2014FullyFunctional}. Also, storing other values allows operations of finding the leftmost maximum of \textproc{excess} and counting the minima of \textproc{excess} within the given range. Although these operations are not required when implementing tree navigational queries in the $\BPb$ representation, they are needed to implement some tree navigational queries in the $\BPo$ representation.

As an example where the additional BP operations are required, let us describe a limited extension of the representation of binary trees to ordinal trees. The $\BPo$ sequences of certain subclasses coincide with the $\BPb$ sequences of compressible subclasses of binary trees. For example, LRM-trees~\cite{Barbay2012LRMTrees} of random permutations correspond to Cartesian trees of random permutations, and ordinal trees with $r$ leaves correspond to Cartesian trees of arrays with $r$ increasing runs. Since we can optimally compress these subclasses of binary trees with hypersuccinct binary trees, we can also optimally compress the corresponding ordinal trees by regarding the $\BPo$ sequence as the $\BPb$ sequence.

Here, we have three options to encode these subclasses of ordinal trees: (a) just applying Theorem~\ref{thm:TC_index_modified} to the ordinal tree, (b) converting the ordinal tree into a binary tree and applying Theorem~\ref{thm:TC_index_modified}, and (c) converting the ordinal tree into a binary tree and applying Theorem~\ref{thm:minimal_design}. 
The space consumption is expected to decrease in the order of (a), (b), and (c). However, (c) does not support tree navigational queries in the $\BPo$ representation. Thus, Theorem~\ref{thm:TC_index_modified} applied to binary trees is expected to be advantageous in this case.

\section{Optimization of Hypersuccinct RMQ}
\label{sec:optimization_RMQ}

\subsection{Space Optimization}

\subsubsection{Motivation}
\label{sec:space_opt_motivation}

Here, we consider the space consumption of the data structures. We first show that there is a difficult trade-off when optimizing space by choosing $B$. We then describe how to tackle the trade-off in the subsequent sections.

We first discuss the implementation of Huffman coding. The implementation uses canonical Huffman codes~\cite{Navarro2016Compact, Schwartz1964Canonical}, which are fast to decode and space-efficient. 
Decoding a canonical Huffman code requires storing a table of the corresponding BP sequences, which consumes $\Order(2^{4B} B)$ bits when encoding micro trees. The reason for taking $B = \lceil (\lg n) / 8 \rceil$ is to suppress the term to $\order(n)$. 

However, $B$ is significantly small when applied to a tree of practical size: if $n = 10^9$, we have $(\lg n) / 8 \sim 3.7$. We can increase $B$ to $(\lg n) / (4 + \epsilon)$ for an arbitrarily small positive number $\epsilon$ to keep the space complexity asymptotically negligible. However, it is still only up to $B = 7$ for $n = 10^9$. Later numerical experiments also confirm that taking a value of $B$ larger than $(\lg n) / (4 + \epsilon)$ leads to non-negligible space consumption. 

Since we consume $\Order(n (\log B) / B)$ bits for storing split ranks and $\Order(n / B)$ bits for the top-tier tree BP, taking a small value of $B$ also leads to a significant consumption of space. When choosing $B$ to optimize space consumption, there is a trade-off between terms $\Order(2^{4B} B)$ and $\Order(n (\log B) / B)$, which are both theoretically negligible but practically significant.

\subsubsection{Encoding Pairs of Split Ranks and Micro Trees}
\label{sec:encoding_pairs_theory}

In this section, we describe a trick to reduce the space consumption of storing the split ranks to address the trade-off: it is effective in practice to encode pairs of split ranks and micro trees with Huffman codes, as presented in Sec.\,\ref{subsec:encoding_pairs_experiment}. This section shows that the asymptotic code length remains the same.

The original algorithm of hypersuccinct binary trees encode micro trees $\mu_1, \dots, \mu_m$ with a Huffman code $H$. The data structure consumes $\sum_{i = 1}^{m} |H(\mu_i)| + \Order(n (\log B) / B)$ bits in total~\cite{Munro2021Hypersuccinct}. Here, we show that the total code length is asymptotically not worsened and remains optimal if we incorporate split ranks into micro trees. 

\begin{theorem}
    Suppose we have a sequence of micro trees $\mu_1, \dots, \mu_m$ and the corresponding split ranks $s_1, \dots, s_m$. Let $H$ be a Huffman code of micro trees in the original scheme and $H'$ be a Huffman code of pairs of micro trees and split ranks. Then, the total code length of $H'$ can be bounded as
    \begin{equation}
        \sum_{i = 1}^{m} |H'((\mu_i, s_i))| \le \sum_{i = 1}^{m} |H(\mu_i)| + \Order\left(\frac{n \log B}{B}\right).
    \end{equation}
\end{theorem}

\begin{proof}
    To prove the theorem, it is sufficient to construct a prefix-free code $D'$ that encodes the pairs of micro trees and split ranks and achieves the same upper bound in the theorem. Then, the optimality of Huffman codes finishes the proof.
    
    To construct a prefix-free code $D'$, we utilize Elias gamma code~\cite{Elias1975Universal}. Elias gamma code $\gamma$ is a prefix-free code that encodes a positive integer $n$ using $2 \lfloor \lg n \rfloor + 1$ bits. The code $\gamma(n)$ is defined as a binary sequence obtained by prepending to the binary representation of $n$ a sequence of zeros that is one shorter than the length of the representation.

    Then, we define $D'$ as $D'((\mu_i, s_i)) = H(\mu_i) \cdot \gamma(s_i)$. Since both $H$ and $\gamma$ are prefix-free, the code $D'$ is also prefix-free. Thus, by the optimality of Huffman codes, we have $\sum_{i = 1}^{m} |H'((\mu_i, s_i))| \le \sum_{i = 1}^{m} |D'((\mu_i, s_i))|$. Also, by summing the equation $|D'((\mu_i, s_i))| = |H(\mu_i)| + 2 \lfloor \lg s_i \rfloor + 1$, it follows that $\sum_{i = 1}^{m} |D'((\mu_i, s_i))| \le \sum_{i = 1}^{m} |H(\mu_i)| + \Order(n (\log B) / B)$. 
\end{proof}

Although the total code length asymptotically remains the same, it significantly reduces the total code length in practice, as shown in Sec.\,\ref{subsec:encoding_pairs_experiment}. 

\subsubsection{Depth-First Arithmetic Code}
\label{subsubsec:depth_first_arithmetic_code}

Although the trick described in the previous section significantly reduces the space consumption of storing the split ranks and mitigates the trade-off, the BP representation of the top-tier tree still consumes $\Order(n / B)$ bits, triggering another trade-off between terms $\Order(2^{4B} B)$ and $\Order(n / B)$.

To eliminate the trade-offs between the space-complexity terms, we can employ a \textit{depth-first arithmetic code}~\cite{Munro2021Hypersuccinct} instead of a Huffman code.
As discussed in Sec.\,\ref{subsec:hypersuccinct_trees}, the code is a tool used in Munro et al.~to show that the Huffman code achieves optimal space consumption. Although the depth-first arithmetic code takes a longer time to decode and lacks universality for tree sources, the depth-first arithmetic code does not need a decoding table and achieves optimal space consumption for random permutations. Thus, by employing the depth-first arithmetic code, we can take an arbitrarily large value of $B$ to reduce space consumption. The drawback of taking a large value of $B$ is that the decoding time becomes long: the average code length of each micro tree is $\Order(B)$ bits and decoding takes time linear in code length. However, we later experimentally show that it remains efficient if we take a large value of $B$. 

\subsection{Time Optimization}

\subsubsection{Branch Pruning and Breadth-First Arithmetic Code}

Sec.\,\ref{subsubsec:depth_first_arithmetic_code} utilizes a depth-first arithmetic code to overcome the trade-offs between the $\order(n)$-bit indexes. Here, we show a branch pruning for \ac{LCA} computation within a micro tree and propose replacing the \ac{DFS} order with the \ac{BFS} order for better pruning.

The decoding procedure of the depth-first arithmetic code simulates \ac{DFS} and obtains the left subtree sizes in the \ac{DFS} order. Thus, when computing the \ac{LCA} on a micro tree, we can check if the current node is the \ac{LCA} while decoding the code; if one of the two nodes appears in the left subtree and the other in the right subtree, the current node is the LCA. This allows pruning the search and terminating the decoding procedure midway.

As a trick for optimizing \ac{RMQ}, we can replace the \ac{DFS} with \ac{BFS} in the design of the depth-first arithmetic code: we call the code a \textit{breadth-first arithmetic code}. Since \ac{BFS} is expected to visit the LCA earlier than \ac{DFS}, decoding is also likely to be terminated sooner by using the \ac{BFS} order. The later numerical experiment in Sec.\,\ref{subsec:depth_first_breadth_first_comparison} matches this expectation. 

\subsubsection{Chunk-Local Ranks}

We discuss another optimization of the hypersuccinct binary trees when the application is limited to RMQ: all the queries required are \textproc{inorder}, \textproc{inorderselect}, and \textproc{lca}. These queries do not refer to the indices of the opening parentheses. Thus, when storing a pair of indices to refer to a closing parenthesis, we can store the chunk-local rank of the closing parenthesis instead of the chunk-local index as the second element. 

The advantage of storing the rank of the closing parenthesis is that we do not need to read the micro-tree BP when executing \textproc{inorder} and \textproc{inorderselect}, which saves decoding the code of the micro tree. When using the chunk-local index, we first scan chunk-wise to find the desired chunk and then read the chunk to find the chunk-local index from the rank of the closing parenthesis within it, as in the $\selectc$ operation; the second part can be skipped if we directly use the closing parenthesis rank.

We note that this optimization is very well-suited for the \textproc{lca} operation because the \textproc{lca} operation requires decoding the micro-tree BP at most once. In particular, among the steps described in Prop.\,\ref{prop:lca_in_minimal}, Step 3 does not decode the micro tree at all since the LCA lies in a micro tree consisting of a single node. Also, when assuming the random permutations and the random inputs of the RMQ, Step 3 is most likely to return the \ac{LCA}. 

Therefore, decoding of the micro-tree BP is performed at most once for each \ac{RMQ} and is mostly skipped for random permutations and random inputs. This is tremendously beneficial when using arithmetic codes and taking a large value of $B$.

\section{Performance Evaluation}
\label{sec:performance_evaluation}

\subsection{Implementation Details}

In this section, we give details of implementation. We implemented the mechanism discussed in Theorem~\ref{thm:minimal_design} and Remark~\ref{rem:actual_impl}. 
For convenience, we restate the implemented data structure in a self-contained manner.
\begin{itemize}
    \item A variable-length cell array of micro-tree codes $D(\mu_1), \dots, D(\mu_m)$ supporting random access.
    \item A data structure that represents the BP sequence $S$ of the top-tier tree. It handles basic operations on the BP sequence $S$ in constant time and spends $\Order(n / B)$ bits~\cite{Navarro2014FullyFunctional}.
    \item An integer array whose $i$-th element contains the split rank of $\mu_i$.
    \item An integer array that stores sparse sampling of the number of closing chunks.
\end{itemize}

Thus, our implementation of hypersuccinct binary trees supports queries in Alg.\,\ref{alg:navigation_on_bp} except for \textproc{preorder} and \textproc{preorderselect}. Also, we prepared \ac{RMQ} data structures utilizing hypersuccinct binary trees and the techniques discussed in Sec.\,\ref{sec:optimization_RMQ}.
We used a simple variant of the range min tree with $\Order(\log n)$ query time. While we employ the design above when we use a depth-first arithmetic code and a breadth-first arithmetic code, we need to modify the design a little when we use a Huffman code: we also need a table for decoding, and we can incorporate an array of split ranks into Huffman codes as discussed in Sec.\,\ref{sec:encoding_pairs_theory}.

We implemented the data structures in C++17, compiled it with GCC 11.4.0, and optimized it with the \texttt{-O3} option. We evaluated the program on a laptop equipped with an 11th Gen Intel(R) Core(TM) i5-1135G7 @ \SI{2.40}{\giga\hertz} with \SI{16}{\giga\byte} memory. 

We explain the parameters which are needed to discuss the later results.
When storing Huffman codes in variable-length cells, we need to store the length of each cell to support random access. However, it consumes $\Order(n (\log B) / B)$ bits, which is significant as the value of $B$ is small. Therefore, we store a constant number of codes together in a single variable-length cell to reduce space consumption. Each micro tree can be retrieved by decoding the previous codes in the same cell and then the code. The number of codes stored in a single cell is basically chosen to be $16$ in implementation; we change the value when evaluating how much space consumption can be reduced. On the other hand, we store every code length when using arithmetic coding.

Also, we need to choose the constant regarding the sparse sampling of closing parentheses: we need to decide the number of chunks to be bundled into a single unit, on which prefix sums are calculated and stored. The constant is basically taken to be $16$; we change the value when evaluating how much space consumption can be reduced.

We then discuss some parameters that are fixed throughout the performance evaluation. The implementation uses the range min-max tree~\cite{Navarro2014FullyFunctional} to represent the top-tier tree. The range min-max tree regards a constant number of bits as a leaf. To reduce the space consumption, we set the constant to $1024$, which is slightly larger than the compromise value between time and space proposed by Ferrada and Navarro~\cite{Ferrada2017ImprovedRMQ}. Optimizing this constant to fit each application is left as future work.

Implementation of variable-length cell arrays uses three-level decomposition. In two-level decomposition, we store element lengths in blocks and then create superblocks that store the prefix sums sampled for a constant number of blocks. In three-level decomposition, superblocks are again decomposed into two blocks as in the dense-pointer technique~\cite{Navarro2016Compact}. These constants are both chosen to be $32$ in implementation. It took less than \SI{0.1}{\micro\second} on average to retrieve a prefix sum, which is negligible in time consumption.

\subsection{Encoding Pairs}
\label{subsec:encoding_pairs_experiment}

In this section, we briefly describe the effect of encoding pairs of micro trees and split ranks with a Huffman code. We set the length of the permutation to $10^9$ and the tree-covering parameter $B$ to $8$. This choice of $B$ minimized the space consumption whether encoding pairs or only micro trees. 

When encoding only the micro trees and separately storing the split ranks as a sequence, the total space consumption of the RMQ data structure was $3.031 n$ bits; the Huffman code sequence, the decoding table, and the sequence of split ranks consumed $1.646 n$ bits, $0.010 n$ bits, and $0.724 n$ bits, respectively. Thus, storage of the micro trees and the split ranks spent $2.380 n$ bits in total. 

When encoding the pairs of micro trees and split ranks, the total space consumption of the RMQ data structure decreased to $2.471 n$ bits; the Huffman code sequence and the decoding table consumed $1.777 n$ bits and $0.043 n$ bits, respectively. Thus, storage of the micro trees and the split ranks spent $1.820 n$ bits in total.
Thus, encoding pairs of micro trees and split ranks reduced the space consumption by more than $0.5 n$ bits in this case.

\subsection{Random RMQ Using Huffman Codes}

\subsubsection{Tree-Covering Parameter}

In this section, we describe the optimization of RMQ using Huffman codes for random permutations.
We first discuss how the value of the parameter $B$ determines the space consumption and the query time when micro trees and split ranks are encoded with Huffman coding. We created the RMQ data structure on a random permutation of length $10^9$ and prepared $10^6$ random queries.

Fig.\,\ref{fig:optimize_B_Huffman} shows the space consumption and the average query times when changing the value of $B$ from $2$ to $11$. The space consumption of the data structure was minimized when $B = 8$, which is approximately $(\lg n) / 4 \sim 7.47$, consuming $2.471 n$ bits; this was considerably larger than the theoretical asymptotic space usage of $1.736n$ bits and was also larger than the worse-case space consumption of $2n$ bits. Considering that the fast and succinct RMQ implementation~\cite{Baumstark2017Practical, Ferrada2017ImprovedRMQ} spends $2.1n$ bits and processes a query in microseconds, this implementation using a Huffman code consumed much more space and processed a query approximately ten times slower.

\begin{figure}[hbtp]
    \centering
    \includegraphics[width=\columnwidth]{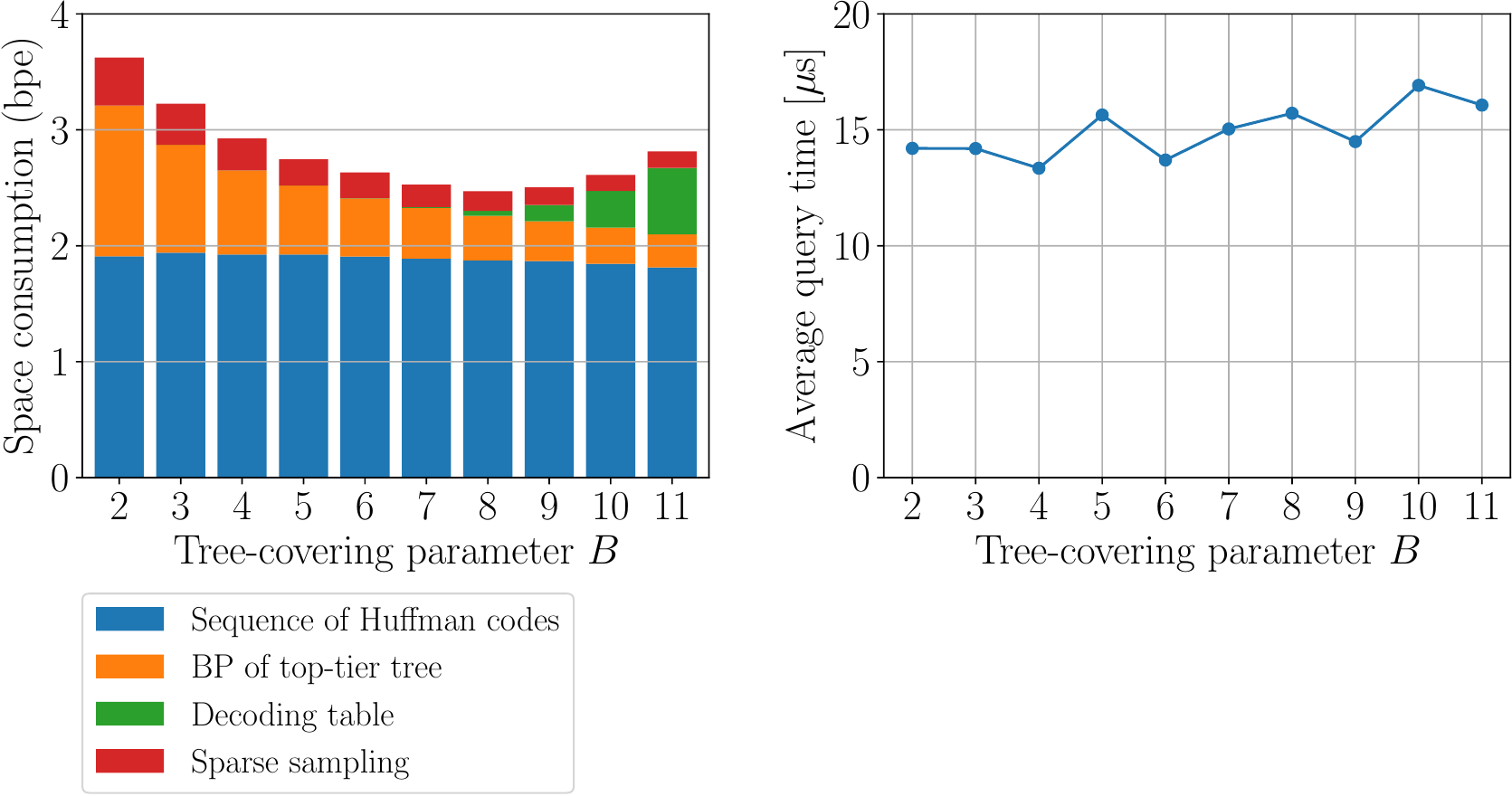}
    \caption{Optimizing the tree-covering parameter $B$ to minimize the space consumption of the RMQ data structure when using a Huffman code to encode micro trees. The array is a random permutation of length $10^9$. Space consumption is expressed as bits per element (bpe).}
    \label{fig:optimize_B_Huffman}
\end{figure}

We note that the similar choice of $B$ was effective when we changed the size $n$ from $10^4$ to $10^9$ by a factor of $10$. The nearest integer to $(\lg n) / 4$ experimentally minimized the space consumption for all the values of $n$ from $10^4$ to $10^8$. However, when $n = 10^9$, the space consumption was minimized by choosing $8$, not the nearest integer to $(\lg n) / 4 \sim 7.47$. 

Although not necessarily true in theory, the results show that the space consumption of the decoding table is negligible even when $B$ is slightly larger than $(\lg n) / 4$. There are two possible reasons for this. First, some BP sequences may not appear in the decoding table, thereby reducing the space consumption below the worst case. Also, due to the small value of $B$, the actual number of micro trees may be small compared to $2^{4B}$, the asymptotic evaluation.

\subsubsection{Sparse Sampling and Code Allocation}

We then consider optimizing the sparse-sampling parameter and the number of Huffman codes allocated in a single cell. We used the same constants for the sparse-sampling parameter and the number of Huffman codes in a cell and varied them from $1$ to $256$ by a factor of $2$. We created the RMQ data structure on a random permutation of length $10^9$ and prepared $10^6$ random queries, with the tree-covering parameter $B$ set to $8$.

Fig.\,\ref{fig:optimize_W} shows the space consumption and the average query time when varying the parameter of sparse sampling and the number of Huffman codes in a cell. Each point is annotated with the parameter. We also show the minimum required space consumption as a vertical dotted line: we only consider the total length of Huffman codes, the top-tier tree BP, and the micro-tree BP sequences, which spend $2.182 n$ bits in total. Choosing $16$ as the parameter seems a good balance between space consumption and query time. We also observe that it is difficult to keep space usage below $2n$ bits using Huffman codes.

\begin{figure}[hbtp]
    \centering
    \includegraphics[width=0.5\columnwidth]{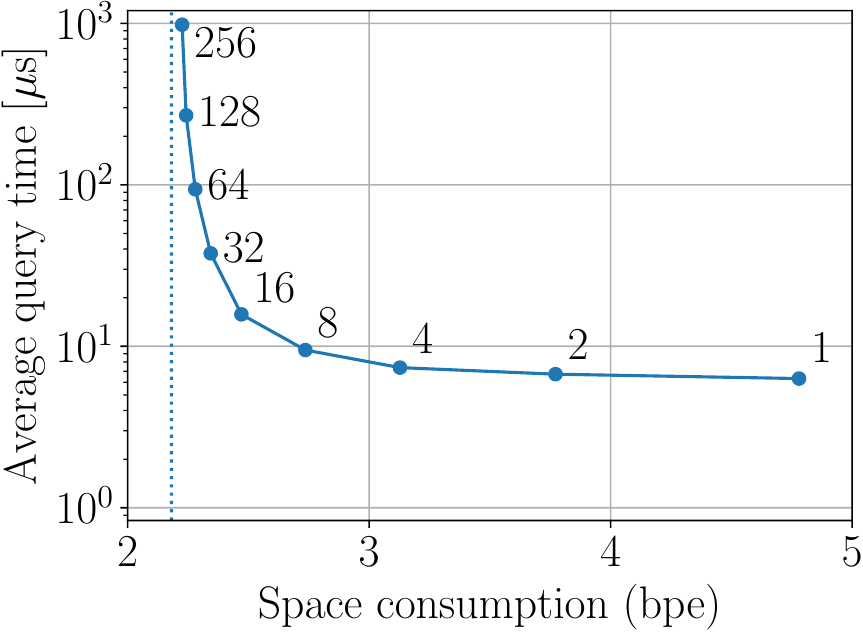}
    \caption{A trade-off between time and space when using a Huffman code. The number of Huffman codes in a single variable-length cell and the number of chunks bundled for sparse sampling are set to the same constant, which is changed from $2^0 = 1$ to $2^8 = 256$ and shown next to each point. The vertical dotted line represents the minimum required space when ignoring the time efficiency.}
    \label{fig:optimize_W}
\end{figure}

\subsection{Random RMQ Using Arithmetic Coding}

\subsubsection{Depth-First and Breadth-First Arithmetic Code}
\label{subsec:depth_first_breadth_first_comparison}

We compare a depth-first arithmetic code and a breadth-first arithmetic code. We also check if pruning for computing \ac{LCA}s is effective. We fixed the array length to $10^9$ and the tree-covering parameter $B$ to $512$.

We first compared the average query time for $10^6$ random queries. For both codes, it took \SI{4.4}{\micro\second} on average without pruning and \SI{4.3}{\micro\second} with pruning. This difference is slight, and we consider this is because the query width is almost always large and the difference in the decoding procedure scarcely affects query time.

Thus, we changed the query width from $2^1$ to $2^{29}$ and generated $10^5$ random queries for each query width. Fig.\,\ref{fig:depth_breadth} shows the average query time for each width. For all codes, smaller query widths result in longer query times, because decoding micro trees is often required. When using branch pruning, another possible reason for the long query times with small query widths is that the decoding process of a micro-tree code is less likely to terminate early in the process. Also for both codes, the branch pruning improves the query time where the query width is small. Without pruning, the breadth-first arithmetic code is slower than the depth-first arithmetic code, probably because converting sequences of left subtree sizes to BP sequences is slower with breadth-first order. However, when pruning is added, the breadth-first arithmetic code becomes faster than the depth-first arithmetic code for queries of small widths, as expected in the algorithm design. We conclude that the breadth-first arithmetic code with branch pruning is the best and use it in the subsequent performance evaluation of \ac{RMQ} data structures.

\begin{figure}[hbtp]
    \centering
    \includegraphics[width=0.6\columnwidth]{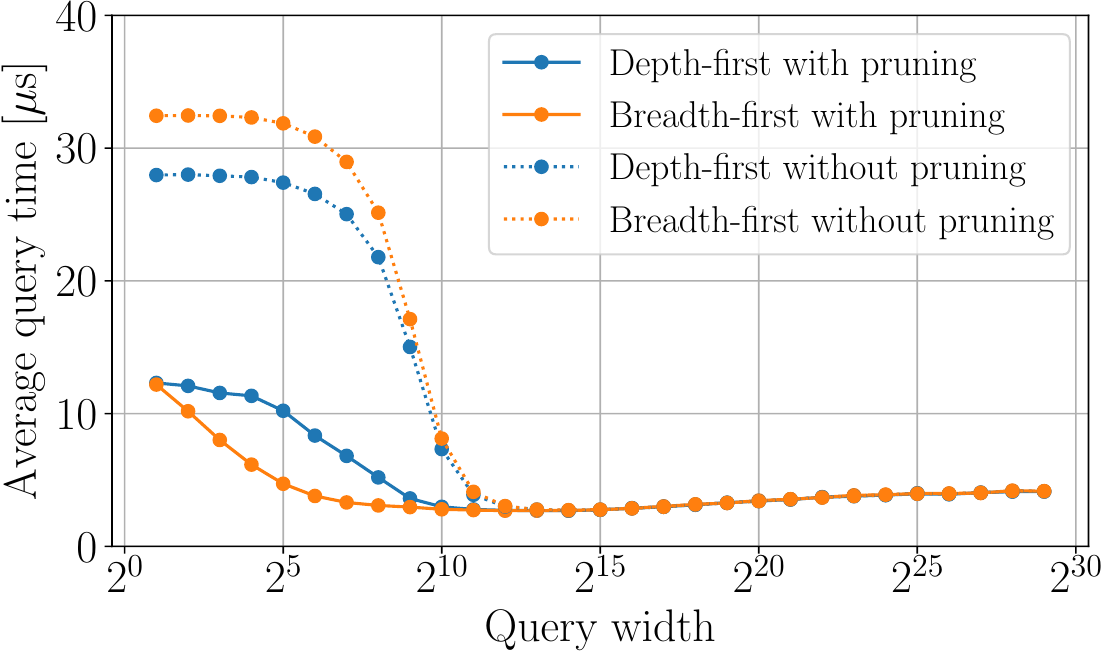}
    \caption{Average query times of several codes when varying query width and the code of micro trees. The array is a random permutation of length $10^9$. The tree-covering parameter $B$ is fixed to $512$.}
    \label{fig:depth_breadth}
\end{figure}

\subsubsection{Tree-Covering Parameter}

We consider optimizing the tree-covering parameter $B$ when using a breadth-first arithmetic code. The value of $B$ is taken as powers of two from $2$ to $2^{20} = 1048576$. The reason for taking $B$ as powers of two is that the number of bits needed to represent each micro-tree size and split rank is not redundant; when taking the logarithm with base two for the number of possible micro-tree sizes and split ranks, there is little gap between the logarithm and its ceiling.

Fig.\,\ref{fig:optimize_B_arithmetic} shows the space consumption and the average query times when using breadth-first arithmetic coding. Most importantly, the space consumption drops to below $2n$ bits when $B$ is $2^8 = 256$ or more, while the average query time is several times slower than the previous results of about a microsecond~\cite{Baumstark2017Practical, Ferrada2017ImprovedRMQ}. For example, if we choose $B$ as $2^8$, $2^9$, and $2^{10}$, it consumes $1.891 n$ bits, $1.822 n$ bits, and $1.784 n$ bits, respectively. We consider the reason for the decrease in average query time as $B$ increases to be that the queries on the top-tier tree probably account for most of the average query time; as $B$ increases, the size of the top-tier tree decreases, which speeds up the queries on the top-tier tree.

\begin{figure}[hbtp]
    \centering
    \includegraphics[width=0.95\columnwidth]{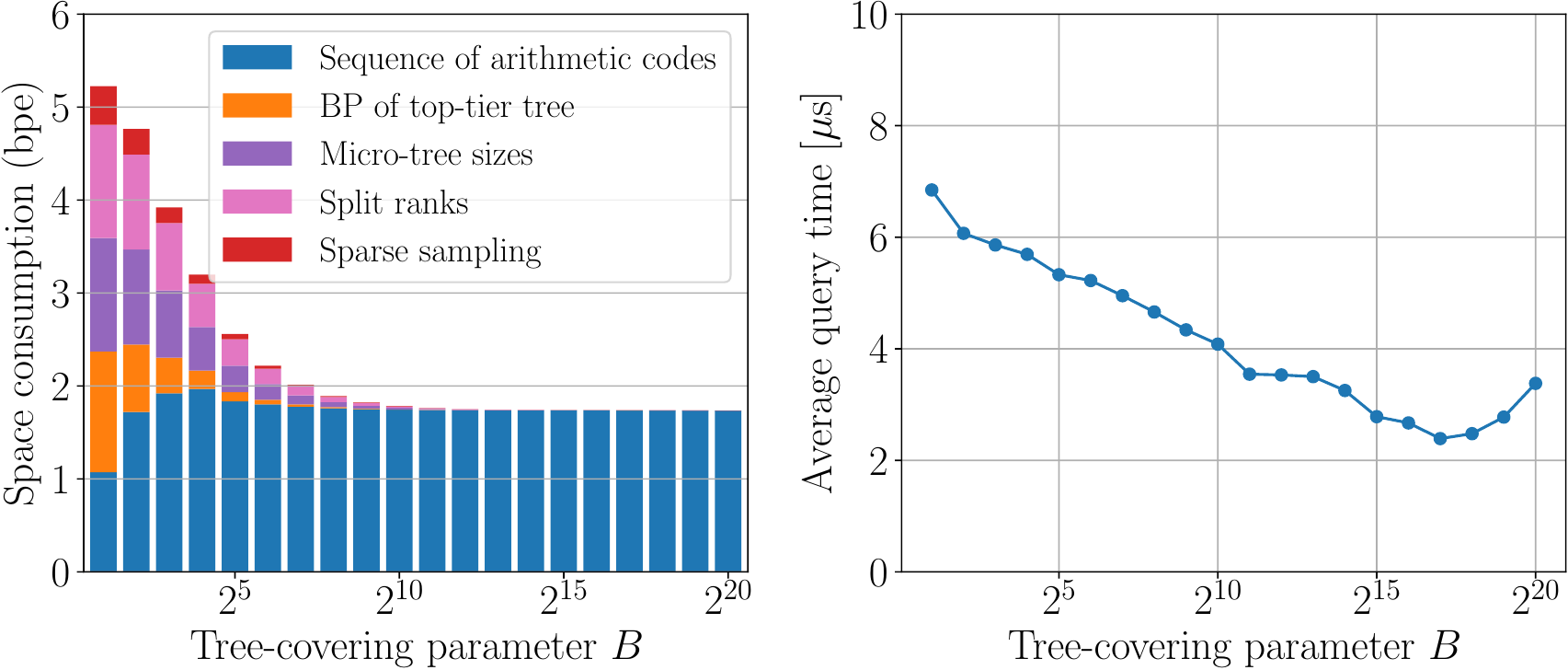}
    \caption{Optimizing the tree-covering parameter $B$ to minimize the space consumption of the RMQ data structure when using a breadth-first arithmetic code and branch pruning. The array is a random permutation of length $10^9$. Space consumption is expressed as bits per element (bpe).}
    \label{fig:optimize_B_arithmetic}
\end{figure}

Although both the space consumption and the average query time decrease as $B$ increases, interpretation of these results requires caution since random queries typically have large widths. In the steps of \ac{LCA} described in Prop.\,\ref{prop:lca_in_minimal}, the \ac{LCA} is often computed in Step 3, in which a micro-tree code is not decoded at all. Although it does not decode a micro-tree code in most cases and runs fast with random queries, we should also consider the query time when the code needs to be decoded. 

To estimate the worst-case query time, we measured the average time required to decode a random breadth-first arithmetic code: for $k = 1, \dots, 20$, we prepared $2^{25-k}$ random permutations of length $2^k$, created the Cartesian trees of the permutations, and measured the average time required to decode breadth-first arithmetic codes encoding the Cartesian trees. Fig.\,\ref{fig:arith_decode_time} shows the average time required to decode a micro tree of $n$ nodes. It took about \SI{0.05}{\micro\second} per node: a micro tree with $256$ nodes can be decoded in about \SI{10}{\micro\second} and $2048$ nodes in about \SI{100}{\micro\second}. Thus, if we choose $B = 128$, the maximum number of micro-tree nodes is about $256$, by which we can estimate the worst-case query time to be \SI{10}{\micro\second}; if $B = 1024$, the worst-case query time can be estimated as \SI{100}{\micro\second}. 

\begin{figure}[hbtp]
    \centering
    \includegraphics[width=0.5\columnwidth]{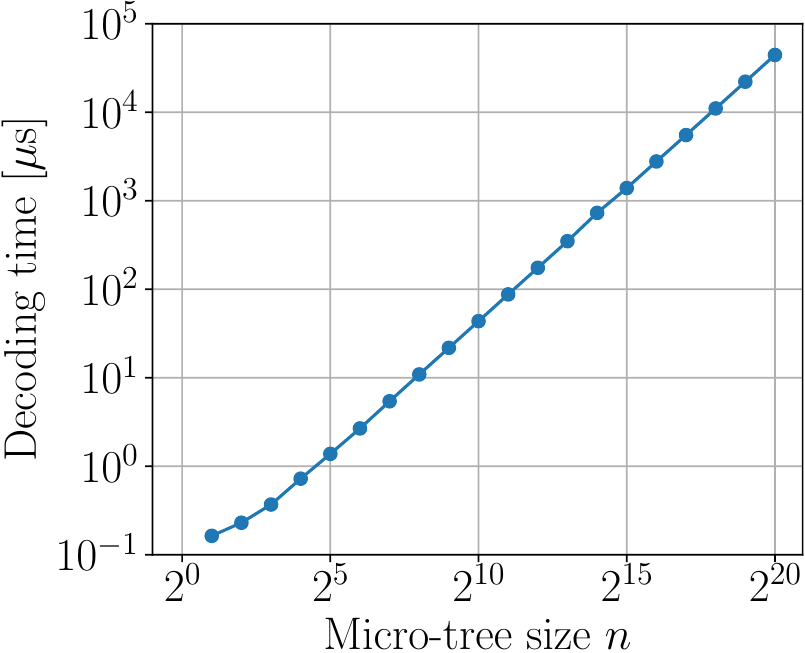}
    \caption{Average time required to decode a breadth-first arithmetic code of a micro tree with $n$ nodes.}
    \label{fig:arith_decode_time}
\end{figure}

The appropriate choice of $B$ depends on various aspects, such as the time allowed as the worst query time, the distribution of query widths, and the desired space consumption. We feel choosing $B$ between $2^6 = 64$ and $2^{10} = 1024$ is a good balance. We later discuss the average query time for each query width with these values of $B$.

\subsection{Random RMQ Comparison}

Here, we compared our implementations with previous implementations by Ferrada and Navarro~\cite{Ferrada2017ImprovedRMQ} and Baumstark et al.~\cite{Baumstark2017Practical} by benchmarking on random permutations. We varied the length of the permutation from $10^4$ to $10^9$ by a factor of $10$ and measured the space consumption and the average query time of the data structures. We prepared the RMQ data structure using a Huffman code with $B$ set to the integer that minimizes the space consumption, i.e., the nearest integer to $(\lg n) / 4$ for $n = 10^4, \dots, 10^8$ and $8$ for $n = 10^9$. We also created the RMQ data structures using a breadth-first arithmetic code with $B$ varying from $64$ to $1024$ by a factor of $2$. Fig.\,\ref{fig:change_n} shows the comparison of the space consumption and the average query time. The RMQ data structure using a breadth-first arithmetic code uses a nearly constant space per element, regardless of the number of elements. It is a natural consequence of the data-structure design, as the space complexity is almost linear to $n$ when fixing $B$. The data structure using a Huffman code is slow, and we consider that this is due to allocating $16$ codes to a single variable cell, which requires decoding several codes in the same cell to access a code.

\begin{figure}[tb]
    \centering
    \includegraphics[width=\columnwidth]{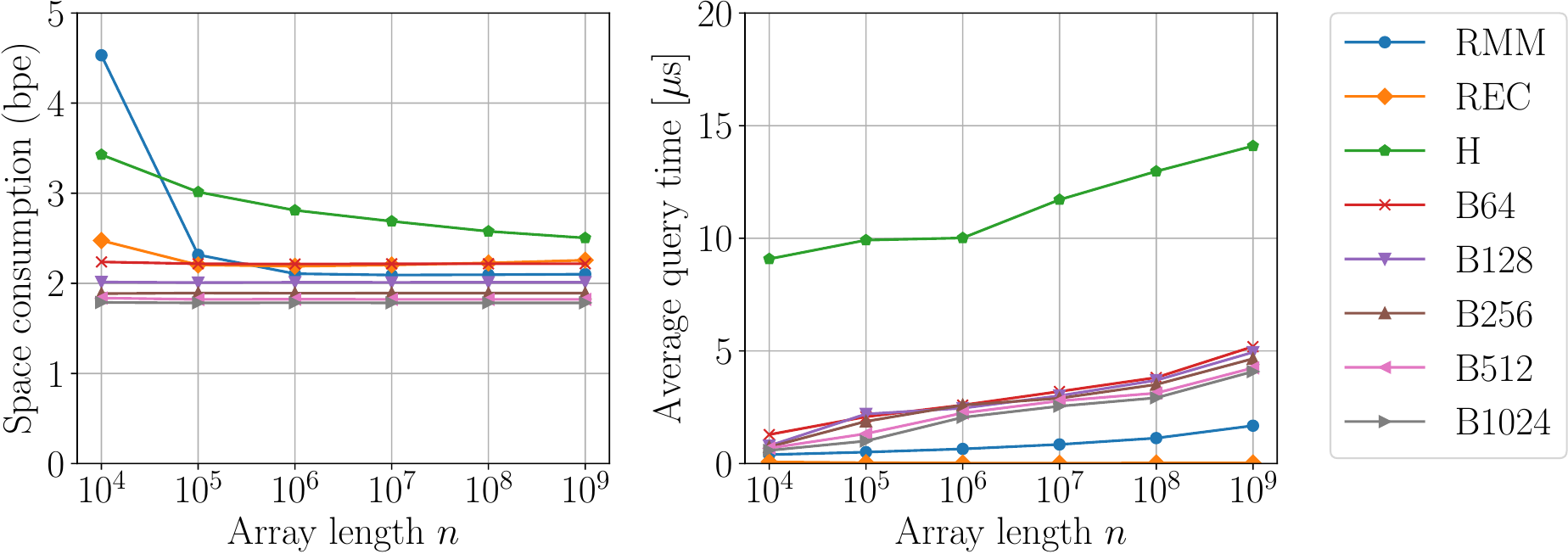}
    \caption{Comparison of the space consumption and the average query time when varying the array length $n$ from $10^4$ to $10^9$. In the legend, RMM refers to the previous implementation by Ferrada and Navarro~\cite{Ferrada2017ImprovedRMQ} and REC refers to the implementation by Baumstark et al.~\cite{Baumstark2017Practical}. The other labels indicate our implementations: H means that micro trees are encoded with Huffman codes, and B followed by an integer means that micro trees are encoded with a breadth-first arithmetic code with tree-covering parameter $B$ set to the shown integer.}
    \label{fig:change_n}
\end{figure}

We also fixed the length of the array as $n = 10^9$, varied the query width from $2^1$ to $2^{29}$ by a factor of $2$, and measured the query time for each query width. We prepared $10^5$ random queries for each width. Fig.\,\ref{fig:hist_time} shows the average query time of each query width. We observe that the graph shapes of breadth-first arithmetic codes are similar even if we vary the tree-covering parameter $B$. Also, we see that breadth-first arithmetic codes with $B \le 512$ work faster than Huffman codes while consuming less space. Although breadth-first arithmetic codes have longer query times when query widths are small, if the original array is available, we expect this disadvantage to be compensated by sequentially scanning the original array.

\begin{figure}[bt]
    \centering
    \includegraphics[width=0.7\columnwidth]{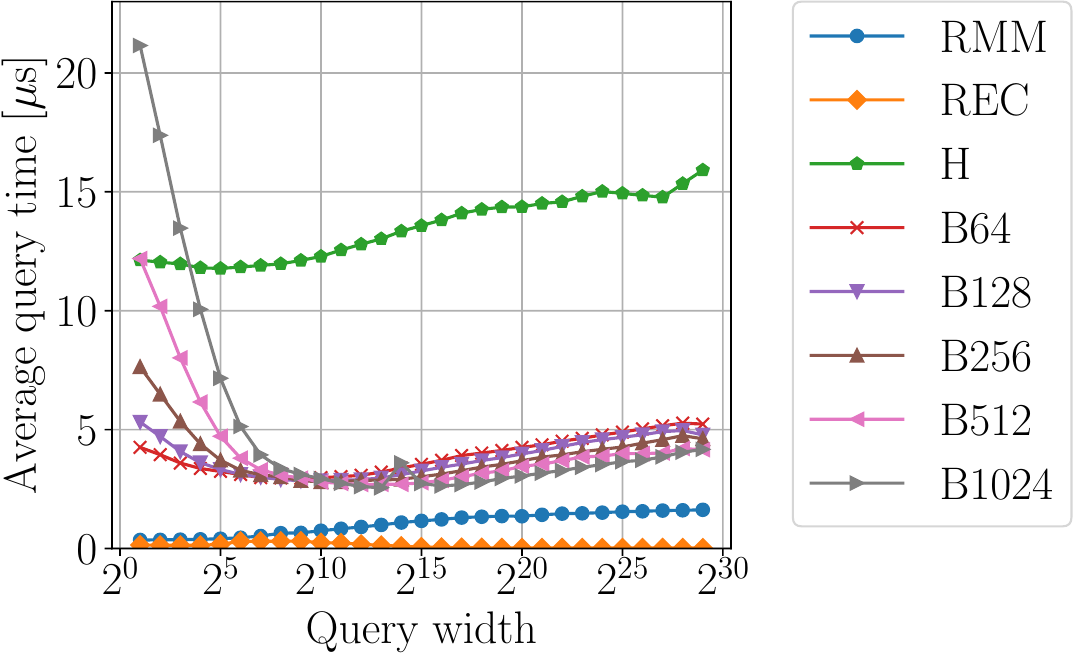}
    \caption{Comparison of the average query time when varying query width. The length of the permutation is set to $10^9$ and each point is an average of $10^5$ queries. The meaning of each label in the legend is described in the caption of Fig.\,\ref{fig:change_n}.}
    \label{fig:hist_time}
\end{figure}

\subsection{RMQ with Runs}
In this section, we evaluate the performance of the \ac{RMQ} data structure when applied to a permutation with increasing runs. Hypersuccinct binary trees encode the Cartesian tree of a permutation of length $n$ with $r$ increasing runs using $2 \lg \binom{n}{r} + \order(n)$ bits~\cite{Munro2021Hypersuccinct}.
We fixed the array length $n$ to $10^8$ and varied $B$ from $5$ to $100$ by steps of $5$. We generated an array with approximately $r'$ runs as follows:
\begin{enumerate}
    \item Shuffle the permutation.
    \item Split the permutation into $r'$ parts: different $r'$ splitting points are uniformly sampled from $n - 1$ possible splitting positions.
    \item Sort each part and concatenate them.
\end{enumerate}
The number of runs in the obtained permutation $r$ is possibly fewer than $r'$, but it works well when $r'$ is small.

\begin{figure}[tb]
    \centering
    \includegraphics[width=0.85\columnwidth]{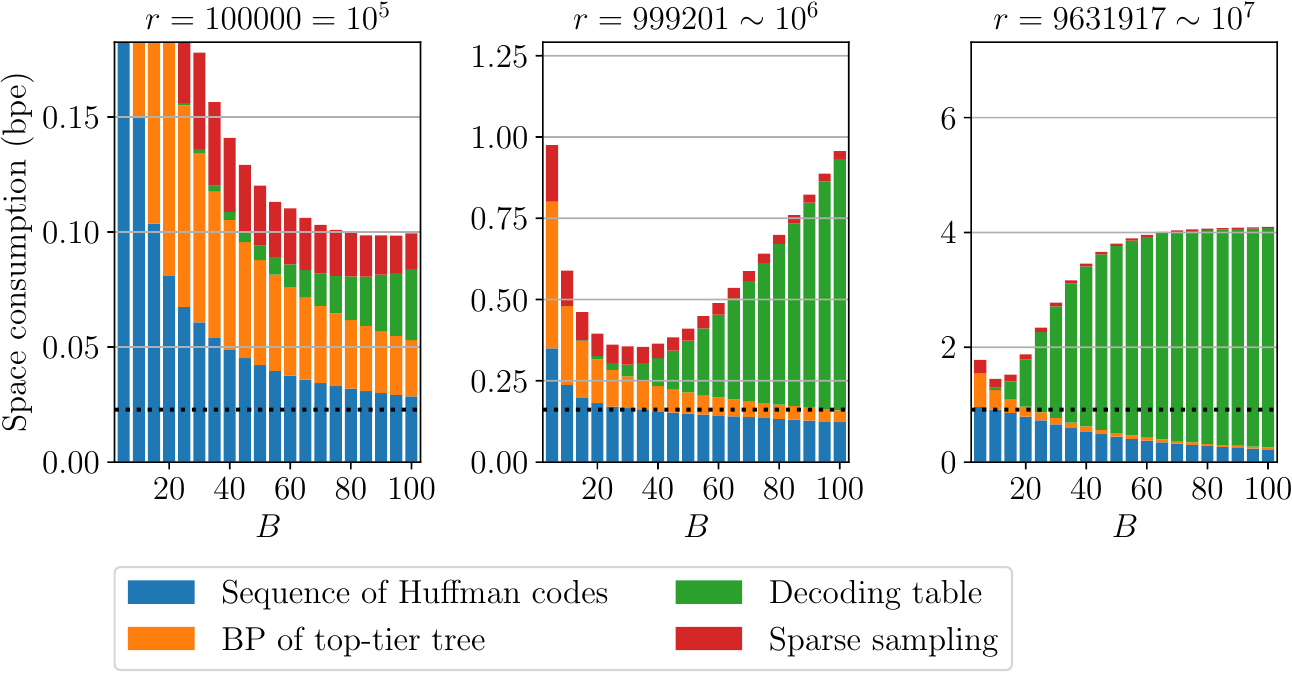}
    \caption{Space consumption breakdown when varying the number of increasing runs $r$ while fixing the array length to $n = 10^8$. Space consumption is expressed as bits per element (bpe). The dotted lines indicate $2 \lg \binom{n}{r} / n$, which is the asymptotical upper bound of the space consumption of the Huffman code sequence.} 
    \label{fig:rmq_runs}
\end{figure}

Setting $r' = 10^5$, $10^6$, and $10^7$, the generated permutation had $100000$, $999201$, and $9631917$ runs, respectively. We created the \ac{RMQ} data structures for these permutations and evaluated the memory performance. Fig.\,\ref{fig:rmq_runs} shows the results. The right plot of the figure shows that when $r \sim 0.1n$, i.e., when the average length of the run is about 10, the space consumption can be less than $2n$ bits: choosing $B = 10$ results in a space consumption of $1.449n$ bits. When $r \sim 0.01 n$, the space consumption takes its minimum value of $0.354n$ bits by choosing $B = 35$. When $r = 0.001n$, the space consumption takes its minimum value of $0.098n$ bits at $B = 95$. Although these values are a few times larger than $2 \lg \binom{n}{r}$, they are significantly smaller than $2n$. As for time efficiency, we prepared $10^6$ random queries and found that the average query time is \SI{14}{\micro\second}.

\subsection{LCP}
As a practical setting, we consider the numerical experiment by Baumstark et al.~\cite{Baumstark2017Practical}: \ac{DFS} on suffix trees by using RMQ data structures over the \ac{LCP} array of real-world texts retrieved from \textit{Pizza\&Chilli} corpus~\cite{Ferragina2009Compressed}.

The results are shown in Table~\ref{tab:suffix_tree_traversal}. Although there is no theoretical guarantee, our average-case optimal \ac{RMQ} data structures tailored to random permutations spend less than $2n$ bits for the DNA text. On the other hand, our RMQ data structures consume more space for other texts. As for query time, our RMQ implementations are considerably slower than previous implementations, probably due to the large number of queries with small query widths during the traversal.

\begin{table}[tb]
    \centering
    \caption{Average space consumption and query time of RMQ data structures over the LCP arrays of texts \textit{Pizza\&Chilli} corpus~\cite{Ferragina2009Compressed}. It simulates \ac{DFS} on the suffix trees by \ac{RMQ}s. 
    The meaning of each label in the leftmost column is described in the caption of Fig.\,\ref{fig:change_n}.}
    \label{tab:suffix_tree_traversal}
    \begin{tabularx}{\textwidth}{l*{8}{>{\raggedleft\arraybackslash}X}} \toprule
        RMQs & \multicolumn{4}{c}{Average space consumption (bpe)} & \multicolumn{4}{c}{Average query time [\si{\nano\second}]} \\ \cmidrule(rl){2-5} \cmidrule(rl){6-9}
	       & dblp  & dna   & english & sources & dblp  & dna   & english & sources \\ \midrule
	RMM   & 2.102 & 2.117 & 2.113   & 2.168   & 202   & 200   & 203     & 215     \\
	REC   & 2.229 & 2.230 & 2.231   & 2.242   & 61    & 61    & 63      & 68      \\
	H     & 2.402 & 2.428 & 2.490   & 2.574   & 10051 & 10047 & 10081   & 10388   \\
	B64   & 2.517 & 2.251 & 2.462   & 2.496   & 2324  & 2323  & 2260    & 2344    \\
	B128  & 2.336 & 2.051 & 2.300   & 2.340   & 2657  & 2690  & 2561    & 2640    \\
	B256  & 2.232 & 1.936 & 2.208   & 2.244   & 3301  & 3320  & 3155    & 3346    \\
	B512  & 2.174 & 1.872 & 2.161   & 2.192   & 4657  & 4642  & 4390    & 4676    \\
	B1024 & 2.143 & 1.835 & 2.136   & 2.169   & 7375  & 7397  & 6948    & 7443    \\
        \bottomrule
    \end{tabularx}
\end{table}

\subsection{Tree Navigational Queries}
\label{sec:tree_navigational_queries_result}

In this section, we evaluate the query times of the hypersuccinct binary trees. We implemented the mechanism discussed in Theorem~\ref{thm:minimal_design} and Remark~\ref{rem:actual_impl}, which supports functions in Alg.\,\ref{alg:navigation_on_bp} except for the first two functions \textproc{preorder} and \textproc{preorderselect}. 

To avoid decoding the same micro tree multiple times, a node is designed to eagerly evaluate the micro tree to which it belongs, i.e., it decodes and retains the micro tree when it is created as an instance. Thus, the branch pruning is not employed when computing LCAs. We encode micro trees with a Huffman code and a depth-first arithmetic code; a breadth-first arithmetic code is not included because decoding the code is slower than decoding a depth-first arithmetic code, as already shown in Fig.\,\ref{fig:depth_breadth}.

We prepared a random permutation of length $10^9$ and created the Cartesian tree of the permutation. We also generated $10^6$ random inputs for each query. Table~\ref{tab:navigation_times} summarizes the average query time with different codes for micro trees. Most of the operations finished in microseconds. The reason the \textproc{inorderselect} function was slow compared to the other operations is that the selected node eagerly decodes a micro tree to which it belongs. When using Huffman codes, \textproc{inorder} was a little slow because of slow random access to a micro tree by allocating several codes to a single variable-length cell. As discussed in Sec.\,\ref{subsubsec:minimal_design}, the \textproc{subtreesize} function internally calls \textproc{inorder} twice, thus it was also slow when using a Huffman code.

\begin{table}[tb]
    \centering
    \caption{Average query time with different codes for encoding micro trees. In the table, H indicates a Huffman code, and D followed by an integer indicates a depth-first arithmetic code with tree-covering parameter $B$ set to the shown integer.}
    \label{tab:navigation_times}
    \begin{tabularx}{\textwidth}{l*{6}{>{\raggedleft\arraybackslash}X}} \toprule
        Navigational operations & \multicolumn{6}{c}{Average query time with varying micro-tree codes [\si{\nano\second}]} \\ \cmidrule(rl){2-7}
        & H & D64 & D128 & D256 & D512 & D1024 \\ \midrule
	\textproc{inorder}             & 3442 & 676  & 625  & 617   & 583   & 560   \\
	\textproc{inorderselect}       & 5525 & 4780 & 8014 & 14336 & 26354 & 50607 \\
	\textproc{root}                & 852  & 594  & 612  & 576   & 481   & 529   \\
	\textproc{parent}              & 334  & 96   & 96   & 110   & 131   & 154   \\
	\textproc{leftchild}           & 174  & 86   & 90   & 105   & 128   & 157   \\
	\textproc{rightchild}          & 372  & 84   & 78   & 90    & 102   & 114   \\
	\textproc{isleaf}              & 525  & 117  & 107  & 112   & 124   & 124   \\
	\textproc{childlabel}          & 314  & 72   & 69   & 80    & 97    & 114   \\
	\textproc{leftmostdescendant}  & 277  & 104  & 103  & 115   & 134   & 153   \\
	\textproc{rightmostdescendant} & 614  & 324  & 344  & 372   & 411   & 445   \\
	\textproc{subtreesize}         & 7028 & 1547 & 1518 & 1554  & 1798  & 1627  \\
	\textproc{isancestor}          & 579  & 275  & 284  & 308   & 341   & 378   \\
	\textproc{lca}                 & 3686 & 2214 & 2318 & 2016  & 1968  & 1995  \\
        \bottomrule
    \end{tabularx}
\end{table}

\section{Conclusion}
\label{sec:conclusion}

In this paper, we proposed a simple representation of tree covering in the BP sequence for both ordinal trees and binary trees. Utilizing the representation, we presented several efficient designs of succinct data structures for trees. Our designs not only reflect the hierarchy but also isolate micro trees as a sequence in the data-structure design, thus enabling compressing micro trees with an arbitrary encoding. We believe the representation can be widely utilized for designing practically memory-efficient data structures based on tree covering. 

We also addressed the implementation of average-case optimal RMQ data structures with hypersuccinct trees. By leveraging our proposed scheme and optimizing the RMQ data structures, our RMQ implementations used less than $2n$ bits and processed queries in a practical time on several settings of the performance evaluation.

Future work includes the method for determining the value of $B$. We manually adjusted the tree-covering parameter $B$ in the present study, but it was unclear how to appropriately choose $B$ when applied to arrays with increasing runs.
Also, the original array may be accessible in some applications, and it would be interesting to design a variant of our work for that case. Implementation and performance evaluation of the mechanism in Theorem~\ref{thm:TC_index_modified} for binary and ordinal trees is another future issue.

\newpage
\bibliographystyle{plainurl}
\bibliography{bib2doi.bib}

\begin{thebibliography}{10}

\bibitem{Arroyuelo2010SuccinctTrees}
Diego Arroyuelo, Rodrigo C{\'{a}}novas, Gonzalo Navarro, and Kunihiko Sadakane.
\newblock Succinct trees in practice.
\newblock In Guy~E. Blelloch and Dan Halperin, editors, {\em Proceedings of the Twelfth Workshop on Algorithm Engineering and Experiments, {ALENEX} 2010, Austin, Texas, USA, January 16, 2010}, pages 84--97. {SIAM}, 2010.
\newblock \href {https://doi.org/10.1137/1.9781611972900.9} {\path{doi:10.1137/1.9781611972900.9}}.

\bibitem{Barbay2012LRMTrees}
J{\'{e}}r{\'{e}}my Barbay, Johannes Fischer, and Gonzalo Navarro.
\newblock Lrm-trees: Compressed indices, adaptive sorting, and compressed permutations.
\newblock {\em Theor. Comput. Sci.}, 459:26--41, 2012.
\newblock \href {https://doi.org/10.1016/j.tcs.2012.08.010} {\path{doi:10.1016/j.tcs.2012.08.010}}.

\bibitem{Baumstark2017Practical}
Niklas Baumstark, Simon Gog, Tobias Heuer, and Julian Labeit.
\newblock Practical range minimum queries revisited.
\newblock In Costas~S. Iliopoulos, Solon~P. Pissis, Simon~J. Puglisi, and Rajeev Raman, editors, {\em 16th International Symposium on Experimental Algorithms (SEA 2017)}, volume~75 of {\em Leibniz International Proceedings in Informatics (LIPIcs)}, pages 12:1--12:16, Dagstuhl, Germany, 2017. Schloss Dagstuhl -- Leibniz-Zentrum f{\"u}r Informatik.
\newblock URL: \url{https://drops.dagstuhl.de/entities/document/10.4230/LIPIcs.SEA.2017.12}, \href {https://doi.org/10.4230/LIPIcs.SEA.2017.12} {\path{doi:10.4230/LIPIcs.SEA.2017.12}}.

\bibitem{Benoit2005DFUDS}
David Benoit, Erik~D. Demaine, J.~Ian Munro, Rajeev Raman, Venkatesh Raman, and S.~Srinivasa Rao.
\newblock Representing trees of higher degree.
\newblock {\em Algorithmica}, 43(4):275--292, Dec 2005.
\newblock \href {https://doi.org/10.1007/s00453-004-1146-6} {\path{doi:10.1007/s00453-004-1146-6}}.

\bibitem{ChakrabortyJSS21}
Sankardeep Chakraborty, Seungbum Jo, Kunihiko Sadakane, and Srinivasa~Rao Satti.
\newblock Succinct data structures for series-parallel, block-cactus and 3-leaf power graphs.
\newblock In {\em COCOA 2021, Tianjin, China, December 17-19, 2021, Proceedings}, volume 13135 of {\em Lecture Notes in Computer Science}, pages 416--430. Springer, 2021.
\newblock \href {https://doi.org/10.1007/978-3-030-92681-6\_33} {\path{doi:10.1007/978-3-030-92681-6\_33}}.

\bibitem{Chakraborty2021CliqueWidth}
Sankardeep Chakraborty, Seungbum Jo, Kunihiko Sadakane, and Srinivasa~Rao Satti.
\newblock Succinct data structures for small clique-width graphs.
\newblock In {\em 2021 Data Compression Conference (DCC)}, pages 133--142, 2021.
\newblock \href {https://doi.org/10.1109/DCC50243.2021.00021} {\path{doi:10.1109/DCC50243.2021.00021}}.

\bibitem{Chen2008LZ}
Gang Chen, Simon~J. Puglisi, and W.~F. Smyth.
\newblock Lempel--ziv factorization using less time {\&} space.
\newblock {\em Mathematics in Computer Science}, 1(4):605--623, Jun 2008.
\newblock \href {https://doi.org/10.1007/s11786-007-0024-4} {\path{doi:10.1007/s11786-007-0024-4}}.

\bibitem{Davoodi2014EncodingRangeMinima}
Pooya Davoodi, Gonzalo Navarro, Rajeev Raman, and S~Srinivasa Rao.
\newblock Encoding range minima and range top-2 queries.
\newblock {\em Philos. Trans. A Math. Phys. Eng. Sci.}, 372(2016):20130131, may 2014.
\newblock \href {https://doi.org/10.1098/rsta.2013.0131} {\path{doi:10.1098/rsta.2013.0131}}.

\bibitem{Davoodi2017SuccinctBinaryTree}
Pooya Davoodi, Rajeev Raman, and Srinivasa~Rao Satti.
\newblock On succinct representations of binary trees.
\newblock {\em Mathematics in Computer Science}, 11(2):177--189, Jun 2017.
\newblock \href {https://doi.org/10.1007/s11786-017-0294-4} {\path{doi:10.1007/s11786-017-0294-4}}.

\bibitem{Elias1975Universal}
Peter Elias.
\newblock Universal codeword sets and representations of the integers.
\newblock {\em IEEE Transactions on Information Theory}, 21(2):194--203, 1975.
\newblock \href {https://doi.org/10.1109/TIT.1975.1055349} {\path{doi:10.1109/TIT.1975.1055349}}.

\bibitem{Farzan2014UniformParadigm}
Arash Farzan and J.~Ian Munro.
\newblock A uniform paradigm to succinctly encode various families of trees.
\newblock {\em Algorithmica}, 68(1):16--40, Jan 2014.
\newblock \href {https://doi.org/10.1007/s00453-012-9664-0} {\path{doi:10.1007/s00453-012-9664-0}}.

\bibitem{Ferrada2017ImprovedRMQ}
H{\'{e}}ctor Ferrada and Gonzalo Navarro.
\newblock Improved range minimum queries.
\newblock {\em J. Discrete Algorithms}, 43:72--80, 2017.
\newblock \href {https://doi.org/10.1016/j.jda.2016.09.002} {\path{doi:10.1016/j.jda.2016.09.002}}.

\bibitem{Ferragina2009Compressed}
Paolo Ferragina, Rodrigo Gonz\'{a}lez, Gonzalo Navarro, and Rossano Venturini.
\newblock Compressed text indexes: From theory to practice.
\newblock {\em ACM J. Exp. Algorithmics}, 13, feb 2009.
\newblock \href {https://doi.org/10.1145/1412228.1455268} {\path{doi:10.1145/1412228.1455268}}.

\bibitem{Fischer2011Space}
Johannes Fischer and Volker Heun.
\newblock Space-efficient preprocessing schemes for range minimum queries on static arrays.
\newblock {\em SIAM Journal on Computing}, 40(2):465--492, 2011.
\newblock \href {https://arxiv.org/abs/https://doi.org/10.1137/090779759} {\path{arXiv:https://doi.org/10.1137/090779759}}, \href {https://doi.org/10.1137/090779759} {\path{doi:10.1137/090779759}}.

\bibitem{Golin2016Encoding}
Mordecai~J. Golin, John Iacono, Danny Krizanc, Rajeev Raman, Srinivasa~Rao Satti, and Sunil~M. Shende.
\newblock Encoding 2d range maximum queries.
\newblock {\em Theor. Comput. Sci.}, 609:316--327, 2016.
\newblock \href {https://doi.org/10.1016/j.tcs.2015.10.012} {\path{doi:10.1016/j.tcs.2015.10.012}}.

\bibitem{10.1145/2656332}
Roberto Grossi and Giuseppe Ottaviano.
\newblock Fast compressed tries through path decompositions.
\newblock {\em {ACM} J. Exp. Algorithmics}, 19(1), jan 2014.
\newblock \href {https://doi.org/10.1145/2656332} {\path{doi:10.1145/2656332}}.

\bibitem{He2014Framework}
Meng He, J.~Ian Munro, and Gelin Zhou.
\newblock A framework for succinct labeled ordinal trees over large alphabets.
\newblock {\em Algorithmica}, 70(4):696--717, Dec 2014.
\newblock \href {https://doi.org/10.1007/s00453-014-9894-4} {\path{doi:10.1007/s00453-014-9894-4}}.

\bibitem{Hsu2013Completion}
Bo{-}June~Paul Hsu and Giuseppe Ottaviano.
\newblock Space-efficient data structures for top-\emph{k} completion.
\newblock In Daniel Schwabe, Virg{\'{\i}}lio A.~F. Almeida, Hartmut Glaser, Ricardo Baeza{-}Yates, and Sue~B. Moon, editors, {\em 22nd International World Wide Web Conference, {WWW} '13, Rio de Janeiro, Brazil, May 13-17, 2013}, WWW '13, pages 583--594, New York, NY, USA, 2013. International World Wide Web Conferences Steering Committee / {ACM}.
\newblock \href {https://doi.org/10.1145/2488388.2488440} {\path{doi:10.1145/2488388.2488440}}.

\bibitem{Jacobson1989Space}
Guy Jacobson.
\newblock Space-efficient static trees and graphs.
\newblock In {\em 30th Annual Symposium on Foundations of Computer Science}, pages 549--554, 1989.
\newblock \href {https://doi.org/10.1109/SFCS.1989.63533} {\path{doi:10.1109/SFCS.1989.63533}}.

\bibitem{Kieffer2009Structural}
John~C. Kieffer, Enhui Yang, and Wojciech Szpankowski.
\newblock Structural complexity of random binary trees.
\newblock In {\em 2009 IEEE International Symposium on Information Theory}, pages 635--639, 2009.
\newblock \href {https://doi.org/10.1109/ISIT.2009.5205704} {\path{doi:10.1109/ISIT.2009.5205704}}.

\bibitem{Munro2021Hypersuccinct}
J.~Ian Munro, Patrick~K. Nicholson, Louisa~Seelbach Benkner, and Sebastian Wild.
\newblock {Hypersuccinct Trees - New Universal Tree Source Codes for Optimal Compressed Tree Data Structures and Range Minima}.
\newblock In Petra Mutzel, Rasmus Pagh, and Grzegorz Herman, editors, {\em 29th Annual European Symposium on Algorithms (ESA 2021)}, volume 204 of {\em Leibniz International Proceedings in Informatics (LIPIcs)}, pages 70:1--70:18, Dagstuhl, Germany, 2021. Schloss Dagstuhl -- Leibniz-Zentrum f{\"u}r Informatik.
\newblock URL: \url{https://drops.dagstuhl.de/entities/document/10.4230/LIPIcs.ESA.2021.70}, \href {https://doi.org/10.4230/LIPIcs.ESA.2021.70} {\path{doi:10.4230/LIPIcs.ESA.2021.70}}.

\bibitem{Munro2001Succinct}
J.~Ian Munro and Venkatesh Raman.
\newblock Succinct representation of balanced parentheses and static trees.
\newblock {\em SIAM Journal on Computing}, 31(3):762--776, 2001.
\newblock \href {https://arxiv.org/abs/https://doi.org/10.1137/S0097539799364092} {\path{arXiv:https://doi.org/10.1137/S0097539799364092}}, \href {https://doi.org/10.1137/S0097539799364092} {\path{doi:10.1137/S0097539799364092}}.

\bibitem{Navarro2016Compact}
Gonzalo Navarro.
\newblock {\em Compact Data Structures - {A} Practical Approach}.
\newblock Cambridge University Press, New York, 2016.
\newblock URL: \url{http://www.cambridge.org/de/academic/subjects/computer-science/algorithmics-complexity-computer-algebra-and-computational-g/compact-data-structures-practical-approach?format=HB}.

\bibitem{Navarro2014FullyFunctional}
Gonzalo Navarro and Kunihiko Sadakane.
\newblock Fully functional static and dynamic succinct trees.
\newblock {\em ACM Trans. Algorithms}, 10(3), May 2014.
\newblock \href {https://doi.org/10.1145/2601073} {\path{doi:10.1145/2601073}}.

\bibitem{Otter1948NumberOfTrees}
Richard Otter.
\newblock The number of trees.
\newblock {\em Annals of Mathematics}, 49(3):583--599, 1948.
\newblock URL: \url{http://www.jstor.org/stable/1969046}, \href {https://doi.org/10.2307/1969046} {\path{doi:10.2307/1969046}}.

\bibitem{RamanRS07}
Rajeev Raman, Venkatesh Raman, and Srinivasa~Rao Satti.
\newblock Succinct indexable dictionaries with applications to encoding \emph{k}-ary trees, prefix sums and multisets.
\newblock {\em {ACM} Trans. Algorithms}, 3(4):43, 2007.
\newblock \href {https://doi.org/10.1145/1290672.1290680} {\path{doi:10.1145/1290672.1290680}}.

\bibitem{Sadakane2007TextRetrieval}
Kunihiko Sadakane.
\newblock Succinct data structures for flexible text retrieval systems.
\newblock {\em J. Discrete Algorithms}, 5(1):12--22, 2007.
\newblock \href {https://doi.org/10.1016/j.jda.2006.03.011} {\path{doi:10.1016/j.jda.2006.03.011}}.

\bibitem{Schwartz1964Canonical}
Eugene~S. Schwartz and Bruce Kallick.
\newblock Generating a canonical prefix encoding.
\newblock {\em Commun. {ACM}}, 7(3):166--169, Mar 1964.
\newblock \href {https://doi.org/10.1145/363958.363991} {\path{doi:10.1145/363958.363991}}.

\bibitem{Tsur2018representation}
Dekel Tsur.
\newblock Representation of ordered trees with a given degree distribution, 2018.
\newblock \href {https://arxiv.org/abs/1807.00371} {\path{arXiv:1807.00371}}, \href {https://doi.org/10.48550/arXiv.1807.00371} {\path{doi:10.48550/arXiv.1807.00371}}.

\end{thebibliography}

\end{document}